\documentclass[manuscript,screen]{acmart}
\usepackage{algorithmic}
\usepackage{graphicx}
\usepackage{textcomp}
\usepackage{bmpsize}
\usepackage{xcolor}
\usepackage{lipsum}
\usepackage[ruled,linesnumbered,algosection]{algorithm2e}

\usepackage{color,soul,url}
\usepackage{amsmath,amsfonts}
\usepackage{array,multirow}
\usepackage{amsthm}
\newtheorem{theorem}{Theorem}[section]

  \newtheorem{lemma}[theorem]{Lemma}
  
  \newtheorem{mydef}{Definition}
  
\AtBeginDocument{%
  \providecommand\BibTeX{{%
    \normalfont B\kern-0.5em{\scshape i\kern-0.25em b}\kern-0.8em\TeX}}}

\setcopyright{none}
\renewcommand\footnotetextcopyrightpermission[1]{}

\acmBooktitle{}

\begin{document}

\title{Genomic Data Sharing under Dependent Local Differential Privacy}

\author{Emre Yilmaz}
\email{yilmaze@uhd.edu}
\affiliation{%
  \institution{University of Houston-Downtown}
  \city{Houston}
  \state{Texas}
  \country{USA}
}

\author{Tianxi Ji}
\email{txj116@case.edu}
\affiliation{%
  \institution{Case Western Reserve University}
  \city{Cleveland}
  \state{Ohio}
  \country{USA}}

\author{Erman Ayday}
\email{exa208@case.edu}
\affiliation{%
  \institution{Case Western Reserve University}
  \city{Cleveland}
  \state{Ohio}
  \country{USA}}

\author{Pan Li}
\email{pxl288@case.edu}
\affiliation{%
  \institution{Case Western Reserve University}
  \city{Cleveland}
  \state{Ohio}
  \country{USA}}

\renewcommand{\shortauthors}{Yilmaz, et al.}

\begin{abstract}
 Privacy-preserving genomic data sharing is prominent to increase the pace of genomic research, and hence to pave the way towards personalized genomic medicine. In this paper, we introduce ($\epsilon , T$)-dependent local differential privacy (LDP) for privacy-preserving sharing of correlated data and propose a genomic data sharing mechanism under this privacy definition. We first show that the original definition of LDP is not suitable for genomic data sharing, and then we propose a new mechanism to share genomic data. The proposed mechanism considers the correlations in data during data sharing, eliminates statistically unlikely data values beforehand, and adjusts the probability distributions for each shared data point accordingly. By doing so, we show that we can avoid an attacker from inferring the correct values of the shared data points by utilizing the correlations in the data. By adjusting the probability distributions of the shared states of each data point, we also improve the utility of shared data for the data collector. Furthermore, we develop a greedy algorithm that strategically identifies the processing order of the shared data points with the aim of maximizing the utility of the shared data. Considering the interdependent privacy risks while sharing genomic data, we also analyze the information gain of an attacker about genomes of a donor's family members by observing perturbed data of the genome donor and we propose a mechanism to select the privacy budget (i.e., $\epsilon$ parameter of LDP) of the donor by also considering privacy preferences of her family members. Our evaluation results on a real-life genomic dataset show the superiority of the proposed mechanism compared to the randomized response mechanism (a widely used technique to achieve LDP).
\end{abstract}




\maketitle

\section{Introduction}
\label{sec:intro}

Recent advances in genome sequencing technologies have enabled individuals to access their genome sequences easily, resulting in massive amounts of genomic data. 
On one hand, sharing this massive amount of data is important for the progress of genomics research. Genomic data collected by research laboratories or published in public repositories leads to significant breakthroughs in medicine, including discovery of associations between mutations and diseases.
On the other hand, genomic data contains sensitive information about individuals, such as predisposition to diseases and family relationships. Due to privacy concerns, individuals are generally hesitant to share their genomic data. Therefore, how to facilitate genomic data sharing in a privacy-preserving way is a crucial problem.

One way to preserve privacy in genomic data sharing and analysis is to utilize cryptographic techniques. However, these techniques mostly bring interoperability and scalability problems. Encrypted data can only be used for a limited number of operations and high computation costs decrease the applicability of these techniques for large scale datasets.
Local differential privacy (LDP) is a state-of-the-art definition to preserve the privacy of the individuals in data sharing with an untrusted data collector, and hence it is a promising technology for privacy-preserving sharing of genomic data. Perturbing data before sharing provides plausible deniability for the individuals. However, the original LDP definition does not consider the data correlations. Hence, applying existing LDP-based data sharing mechanisms directly on genomic data makes perturbed data vulnerable to attacks utilizing correlations in the data.

In this work, our goal is to provide privacy guarantees for the shared genomic sequence of a data owner against inference attacks (that utilize correlations in the data) while providing high data utility for the data collector. For that, we develop a new genomic data sharing mechanism by defining a variant of LDP under correlations, named $(\epsilon,T)$-dependent LDP. We use randomized response (RR) mechanism (a commonly used LDP-based data sharing mechanism) as a baseline since the total number of states for each genomic data point is 3 and the RR provides the best utility for such a small number of states~\cite{wang2017locally}. Moreover, RR uses the same set of inputs and outputs without an encoding, which allows the data collector to use perturbed data directly. We show how directly applying RR mechanism is vulnerable to inference attacks and focus on improving its privacy and utility while providing formal privacy guarantees. Thus, we first show how correlations in genomic data can be used by an attacker to infer the original values of perturbed data points. We describe a correlation attack and show how estimation error of the attacker (a commonly used metric to quantify genomic privacy) decreases due to the direct use of RR.  

In the correlation attack, to increase its inference power, the attacker detects (and eliminates) the data values that are not consistent with the other shared values based on correlations. Thus, in the proposed data sharing scheme, we consider such an attack by-design and do not share the values of the shared data points which are inconsistent with the previously shared data points. During sharing of each data point (single nucleotide polymorphism - SNP) with the data collector, the proposed algorithm eliminates a particular value of a shared SNP if the corresponding value of the SNP occurs with negligible probability considering its correlations with the other shared SNPs (to prevent an attacker utilize such statistically unlikely values to infer the actual values of the SNPs). Then, the algorithm adjusts the sharing probabilities for the non-eliminated values of the SNP by normalizing them and making sure that the attacker's distinguishability between each possible values of the SNP is bounded by $e^{\epsilon}$, which achieves $(\epsilon,T)$-dependent LDP.

To improve utility, we introduce new probability distributions (for the shared states of each SNP), such that, for each shared SNP, the probability of deviating significantly from its ``useful values'' is small. Useful values of a SNP depend on how the data collector intends to use the collected SNPs. For this, we focus on genomic data sharing beacons (a system constructed with the aim of providing a secure and systematic way of sharing genomic data) and show how to determine probability distributions for different states of each shared SNP with the aim of maximizing the utility of the collected data (this can easily be extended for other uses of genomic data, such as in statistical databases). In the proposed mechanism, SNPs of a genome donor are processed sequentially. Although the proposed $(\epsilon,T)$-dependent LDP definition is satisfied in any order, the number of eliminated states for each SNP can be different based on the order of processing. Hence, the utility of the shared data changes when the processing order of SNPs changes. We also show how to determine an optimal order of processing (which provides the highest utility) via Markov decision process (MDP) and provide a value iteration-based algorithm to achieve this goal. Furthermore, due to complexity of the optimal algorithm, we propose an efficient greedy algorithm to determine the processing order of the SNPs in the proposed data sharing mechanism. 

Since SNPs of a child are inherited from her parents, genomic data sharing should also consider the privacy of the family members. An attacker can gain information about the SNPs of a donor's family member even though the donor shares her SNPs after perturbing their values using the proposed mechanism. Thus, we also explain how the information gain of an attacker about a family member can be computed in terms of the privacy budget (i.e., $\epsilon$ parameter) of the genome donor. We also propose a basic algorithm to identify the maximum privacy budget (i.e., largest $\epsilon$ value) which can be used by a genome donor to make sure that the privacy budgets of her family members are not violated due to her sharing.

We conduct experiments with a real-life genomic dataset to show the utility and privacy provided by the proposed scheme. Our experimental results show that the proposed scheme provides better privacy and utility than the original randomized response mechanism. We also show that using the proposed greedy algorithm for the order of processing, we improve the utility compared to randomly selecting the order of processed SNPs.

The rest of the paper is organized as follows. We review the related work in Section~\ref{sec:related} and provide the technical preliminaries in Section~\ref{sec:preliminary}. We present the proposed framework in Section~\ref{sec:proposed}. We propose an algorithm for optimal data processing order in Section~\ref{sec:order}. We evaluate the proposed scheme via experiments in Section~\ref{sec:evaluation}. In Section~\ref{sec:disc}, we discuss about considering kinship relationships during data sharing and our assumptions about the attacker. Finally, we conclude the paper in Section~\ref{sec:conclusion}.

\section{Related Work}
\label{sec:related}

In this section, we discuss relevant existing works on genomic privacy along with local differential privacy.

\subsection{Genomic Privacy}

Genomic privacy topic has been recently explored by many researchers~\cite{survey:genomicera}. Several works have studied various inference attacks against genomic data including membership inference~\cite{related:homer,related:wang,related:shringarpureandbastumante} and attribute inference~\cite{genomic:lacks,genomic:highorder,Khodam}. To mitigate these threats, some researchers proposed using cryptographic techniques for privacy-preserving processing of genomic data \cite{related:baldi, related:ermanclinic,related:wangPrivateEditDistance, deuber2019my}. The differential privacy (DP) concept~\cite{privacy:differentialprivacy} has also been used to release summary statistics about genomic data in a privacy-preserving way (to mitigate membership inference attacks) \cite{differential:gwas, differential:gwas_yu, differential:gwas_johnson}. Unlike the existing DP-based approaches, our goal is to share the genomic sequence of an individual, not summary statistics. To share genomic sequences in a privacy-preserving way, Humbert~\emph{et al.} proposed an optimization- based technique that selectively hides portions of shared genomic data to optimize utility by considering privacy constraints~\cite{genomic:reconciling}. However, they do not provide formal privacy guarantees. For the first time, we study the applicability of LDP for genomic data sharing and develop a variant of LDP for correlated data.

\subsection{Local Differential Privacy}
Differential privacy (DP)~\cite{privacy:differentialprivacy} is a concept to preserve the privacy of records in statistical databases while publishing statistical information about the database. Although DP provides strong guarantees for individual privacy, there may be privacy risks for individuals when data is correlated. Several approaches have been proposed \cite{yang2015bayesian, cao2017quantifying, liu2016dependence, song2017pufferfish} in order to protect the privacy of individuals under correlations. Since these works focus on privacy of aggregate data release (e.g., summary statistics about data), they are not suitable for individual data sharing. Local differential privacy (LDP) is a state-of-the-art definition to preserve the privacy of the individuals in data sharing with an untrusted data collector. However, a very limited number of tasks, such as frequency estimation \cite{wang2017locally}, heavy hitters \cite{bassily2017practical}, frequent itemset mining \cite{wang2018locally}, marginal release \cite{cormode2018marginal}, and range queries \cite{cormode2019answering} have been demonstrated under LDP and the accuracy of these tasks are much lower than performing the same task under the central model of differential privacy. Collecting perturbed data from more individuals decreases the accuracy loss due to randomization. Hence, practical usage of LDP-based techniques needs a high number of individuals (data owners), which limits the practicality of LDP-based techniques. To overcome the accuracy loss due to LDP, a shuffling technique has recently been proposed \cite{erlingsson2019amplification, cheu2019distributed}. The main idea of shuffling is to utilize a trusted shuffler which receives the perturbed data from individuals and permutes them before sending to data collector. However, requiring a trusted shuffler also restricts the practical usage of this method. 

Another approach to improve utility of LDP is providing different privacy protection for different inputs. In the original definition of LDP, all inputs are considered as sensitive and the indistinguishability needs to be provided for all inputs. Murakami \emph{et al.} divided the inputs into two groups as sensitive and non-sensitive ones \cite{murakami2019utility}. They introduced the notion of utility-optimized LDP, which provides privacy guarantees for only sensitive inputs. Gu \emph{et al.} \cite{gu2019providing} proposed input-discriminative LDP, which provides distinct protection for each input. However, grouping inputs based on their sensitivity is not realistic in practice due to the subjectivity of sensitivity. In this work, we discriminate the inputs based on their likelihood instead of their sensitivity. We focus on how correlations can be used by an attacker to degrade privacy and how we mitigate such degradation. Hence, we provide indistinguishability between possible states by eliminating the states that are rarely seen in the population using correlations. By doing so, we aim to decrease the information gain of an attacker that uses correlations for inference attacks. Furthermore, both of these works~\cite{murakami2019utility, gu2019providing} aim to improve utility by providing less indistinguishability for non-sensitive data and providing more accurate estimations. In our work, the accuracy does not rely on estimations. Instead, we provide high accuracy by eliminating rare values from both input and output sets. Moreover, we improve the utility by increasing the probability of ``useful values'' considering the intended use of the shared data.

\section{Technical Preliminaries}
\label{sec:preliminary}

In this section, we give brief backgrounds about genomics and LDP.

\subsection{Genomics Background}
\label{subsec:genomics}

The human genome contains approximately 3 billion pairs of nucleotides (A, T, C, or G). Approximately $99.9\%$ of these pairs are identical in all people. When more than 0.5\% of the population does not carry the same nucleotide at a specific position in the genome, this variation is considered as single-nucleotide polymorphism (SNP). More than 100 million SNPs have been identified in humans. For a SNP, the nucleotide which is observed in the majority of the population is called the major allele and the nucleotide which is observed in the minority of the population is called the minor allele. Each person has two alleles for each SNP position, and each of these alleles are inherited from a parent of the individual. Hence, each SNP can be represented by the number of its minor alleles, which can be 0, 1, or 2. In this work, we study the problem of sharing the values of SNPs in a privacy-preserving way. It is shown that SNPs may have pairwise correlations between each other (e.g., linkage disequilibrium~\cite{slatkin2008linkage}). Hence, an attacker can use such correlations to infer the original values of the shared SNPs. Furthermore, since alleles are inherited from parents, sharing SNPs by a genome donor also results in revealing some information about the family members. Thus, privacy of the family members should also be considered while sharing genomic data.

\subsection{Local Differential Privacy}
\label{subsec:ldp}

Local differential privacy (LDP) is a variant of differential privacy that allows to share data with an untrusted party. In LDP settings, there is a data collector who wants to compute statistical information about a population. Each individual shares her data with the data collector after perturbation (randomization). Then, the data collector uses all collected perturbed data to estimate statistics about the population. During data perturbation, the privacy of the individuals are protected by achieving indistinguishability. 

In this work, we adopt the  general definition of
local differential privacy \cite{duchi2013local,kairouz2014extremal}, which is expressed as follows:
\begin{mydef}
A randomized mechanism $A$ satisfies $\epsilon$-local differential privacy if 
\begin{equation*}
    \sup_{y\in\sigma(\mathcal{Y}),d_{\alpha},d_{\beta}\in\mathcal{X}}\frac{\Pr(y|A(x = d_{\alpha}))}{\Pr(y|A(x = d_{\beta}))} \leq e^{\epsilon},
\end{equation*}
where $d_{\alpha}$ and $d_{\beta}$ are two possible values of an element $x$, 
$y$ is the output value, $\mathcal{Y}$ is the collection of all possible output values of $x$, and $\sigma(\mathcal{Y})$ denotes an appropriate $\sigma$-filed on $\mathcal{Y}$.
\label{def:general-ldp}
\end{mydef}
Definition \ref{def:general-ldp} captures a type of plausible-deniability, i.e., 
no matter what input value  of $x$ is released, it is nearly equally as
likely to have come from any of its possible values.

The parameter $\epsilon$ is the privacy budget, which controls the level of privacy. Randomized response (RR) \cite{warner1965randomized} is a mechanism for collecting sensitive information from individuals by providing plausible deniability. Although RR is originally defined for two possible inputs (e.g., yes/no), this mechanism can also be generalized to protect privacy when there are more than two possible states. In generalized randomized response \cite{wang2018locally}, the correct value is shared with probability $p = e^\epsilon / (e^\epsilon + $m$ - 1)$ and each incorrect value is shared with probability $q = 1 / (e^\epsilon + $m$ - 1)$ to achieve $\epsilon$-LDP, where $m$ is the number of states.

\section{Proposed Framework}
\label{sec:proposed}

In this section, we first introduce the problem and explain genomic data sharing with an untrusted data collector by directly applying RR mechanism. We then present a correlation attack that utilizes correlations between SNPs and show the significant decrease in the estimation error of the attacker (a commonly used privacy metric for genomic data) after the attack. Then, we show how to simultaneously improve privacy against the correlation attacks and improve utility for genomic analysis. Finally, we present our proposed genomic data sharing mechanism.

\subsection{Problem Statement}
\label{subsec:randomized_response}

\textbf{System Model.} Figure~\ref{fig:system} shows the overview of the system model and the steps of the proposed framework. We focus on a problem, where genome donors share their genomic data in a privacy-preserving way with a data collector who will use collected data to answer queries about the population. In genomic data sharing scenario, there are $n$ individuals ($I_1,...,I_n$) as genome donors. A genome donor $I_j$ has a sequence of SNPs denoted by $\mathcal{X}^j = \lbrack x_{1}^j, \ldots, x_{l}^j \rbrack$. Since each SNP is represented by the number of minor alleles it carries, each $x_{i}^j$ has a value from the set $\{0, 1, 2\}$. \color{black} Today, individuals can obtain their genomic sequences via various online service providers, such as 23andMe, and they also share their sequences with other service providers or online repositories (e.g., for research purposes). Hence, the proposed system model has real-world applications, where individuals want to preserve privacy of their genomic data when they share their genomic sequences with other service providers and online repositories. 

\begin{figure}
\centering
\includegraphics[width=\columnwidth,keepaspectratio]{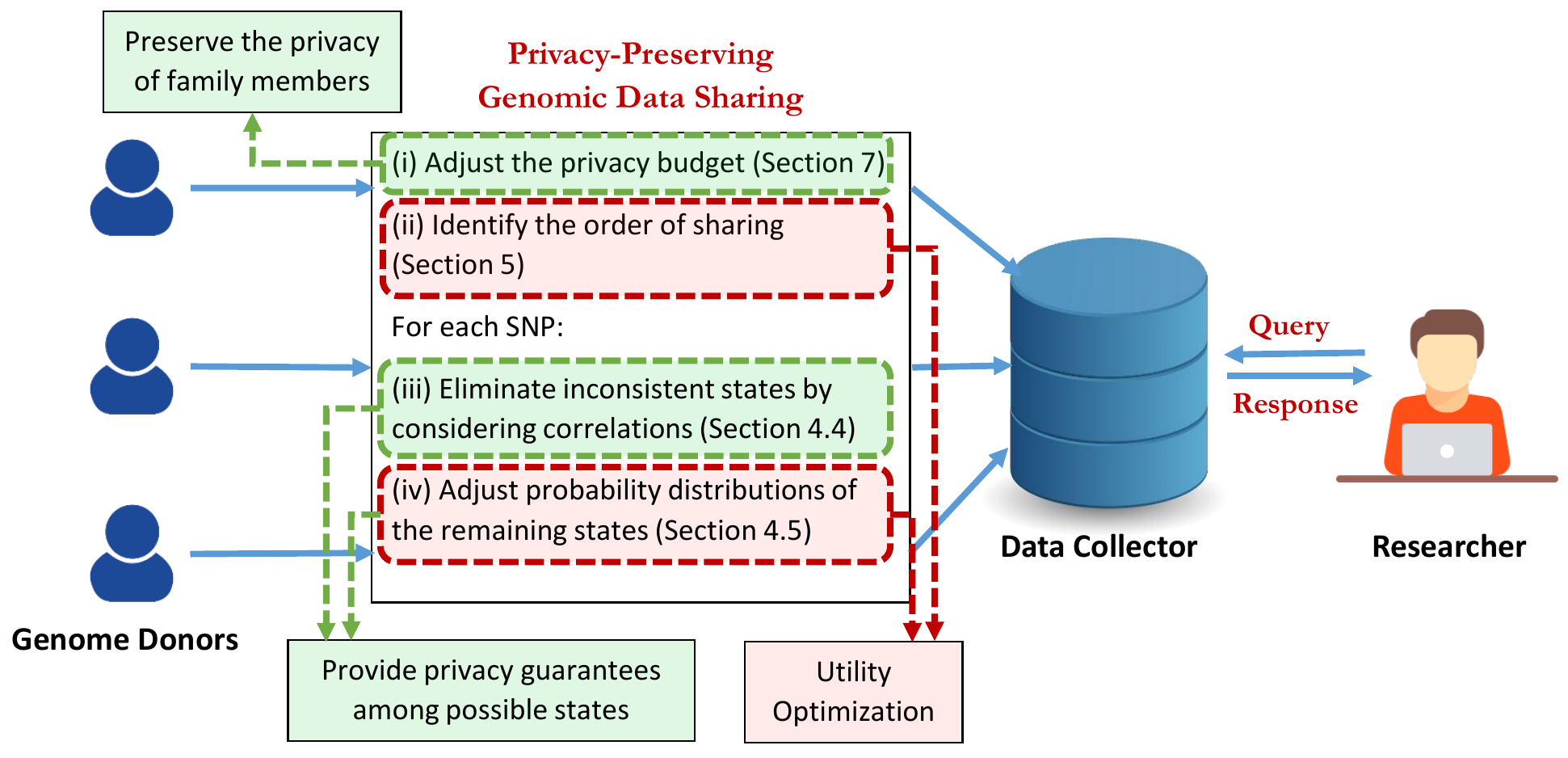}
\caption{System Model.}
\label{fig:system}
\end{figure}

\noindent \textbf{Threat Model.} The data collector is considered as untrusted (i.e., potential attacker). It can share the data directly with another party or use it to answer queries. Hence, we assume the attacker has data shared by all genome donors with the data collector, however, it does not know the original values of any SNPs. In addition, we assume that the attacker knows the pairwise correlations between SNPs (which can be computed using public datasets), the perturbation method, and the privacy budget $\epsilon$. Thus, the attacker can infer whether the shared value of a SNP is equal to its original value using correlations.\vspace{1pt}

\noindent \textbf{Data Utility.} The data collector uses data collected from genome donors to answer queries. Therefore, we define the utility as the accuracy of data collector to answer such queries. For genomic data, typically, the utility of each value of a SNP is different and the utility of a SNP may change depending on the purpose of data collection (e.g., statistical genomic databases, genomic data sharing beacons, or haploinsufficiency studies). Thus, one of our aims is to improve the utility of LDP-based data collection mechanism by considering data utility as a part of the data sharing mechanism.

\vspace{1pt}
\noindent \textbf{Genomic Data Sharing Under Local Differential Privacy.} In \cite{wang2017locally}, several approaches have been explained for estimating frequency of inputs under LDP such as direct encoding, histogram encoding, and unary encoding. As shown in \cite{wang2017locally}, when the size of input set is less than $3e^{\epsilon} + 2$, direct encoding is the best among these approaches. Since the size of input set for genomic data is $3$, we also use direct encoding approach for genomic data sharing. In direct encoding approach, no specific encoding technique is applied to inputs before perturbation and randomized response (RR) mechanism (introduced in Section \ref{subsec:ldp}) is used for perturbing inputs. To apply RR mechanism and achieve $\epsilon$-LDP for genomic data, the value of a SNP is shared correctly with probability $p = e^\epsilon / (e^\epsilon + 2)$ and each incorrect value is shared with probability $q = 1 / (e^\epsilon + 2)$. After receiving perturbed values from $n$ individuals, the data collector estimates the frequency of each input in the population as $\frac{c_i - n \cdot q}{p - q}$, where $c_i$ is the number of individuals who shared $i \in \{0,1,2\}$.

\subsection{Correlation Attack Against LDP-Based Genomic Data Sharing}
\label{subsec:attack}
When multiple data points are shared with the RR mechanism, $\epsilon$-LDP is still guaranteed if the data points are independent. However, it is known that SNPs have pairwise correlations between each other (e.g., linkage disequilibrium~\cite{slatkin2008linkage}). An attacker can use the correlations between SNPs to infer incorrectly or correctly shared SNPs as a result of the RR mechanism.

To show this privacy risk, we consider a correlation attack that can be performed by an attacker in the following. We represent a SNP $i$ as $SNP_i$ and we represent the value of $SNP_i$ for individual $I_j$ as $x_i^j$. 
We assume that all pairwise correlations between SNPs are publicly known. Hence, $\Pr(SNP_i = d_{\alpha}~|~SNP_k = d_{\beta})$ is known by the attacker for any $i, k \in \{1,\ldots,l\}$ and $d_{\alpha}, d_{\beta} \in \{0,1,2\}$. Let $\mathcal{Y}^j = \lbrack y_{1}^j, \ldots,y_{l}^j \rbrack$ be the perturbed data that is shared by $I_j$ with the data collector (potential attacker whose goal is to infer the actual SNP values of the individual). 
Without using the correlations, the attacker's only knowledge about any $x_i^j$ is the probability distribution of randomized response mechanism. However, using the correlations, the attacker can enhance its knowledge about the probability distribution of $x_i^j$ by eliminating the values that are not likely to be observed (i.e., that have low correlation with the other received data points).

To achieve this, for each $SNP_i$ of $I_j$, using all other received data points $\lbrack y_{1}^j, \ldots,y_{l}^j \rbrack$ (except for $y_{i}^j$), the attacker counts the number of inconsistent instances in terms of correlations between different values of $SNP_i$ and all other received data points (i.e., having correlation less than a threshold). 
Let $\tau$ be the correlation threshold of the attacker. The attacker keeps a count for the number of instances for each $SNP_k$ ($k \in \{1,\ldots,l\}, k \neq i$) having $\Pr(x_i^j = 0 ~|~x_k^j = y_{k}^j) < \tau$, $\Pr(x_i^j = 1~|~x_k^j = y_{k}^j) < \tau$, and $\Pr(x_i^j = 2~|~x_k^j = y_{k}^j) < \tau$ as $c_{i,0}^j$, $c_{i,1}^j$, and $c_{i,2}^j$, respectively. If any of these values is greater than or equal to $\gamma \cdot l$ (where $\gamma$ is an attack parameter for the number of inconsistent data points), the attacker eliminates that value in the probability distribution of $x_i^j$ and considers the remaining values for its inference about the correct value of $x_i^j$.

To show the effect of this correlation attack on privacy, we implemented the RR mechanism for genomic data and computed the attacker's estimation error, a metric used in genomic privacy to quantify the distance of the attacker's inferred values from the original data, before and after the attack. Our results (in Figure~\ref{fig:ldpComp}, Section~\ref{subsec:compareLDP}) clearly show \color{black} the vulnerability of directly using RR in genomic data sharing. For instance, when $\epsilon = 1$, the attacker's estimation error decreases from 0.8 to 0.4 after the correlation attack. In general, we observed that the attacker's estimation error (i.e., the privacy of genome donor) decreases approximately 50\% by using this attack strategy. 

\subsection{$(\epsilon,T)$-dependent Local  Differential Privacy}
\label{sec:newDefinition}

To handle data dependency in privacy-preserving data sharing, 
some works, such as \cite{chanyaswad2018mvg,liu2016dependence} extend the definition of traditional differential privacy by considering the correlation between elements in the dataset. However, there is a lack of such variants for local differential privacy models, which hinders the application of LDP-based solutions for privacy-preserving genomic  data sharing. In this paper, inspired by \cite{liu2016dependence}, which handles data dependency by considering the number of elements that can potentially be affected by a single element, we propose the following definition.
\begin{mydef}
An element $x$ in  a dataset $\mathcal{X}$  is said to be $T$-dependent under a correlation model  (denoted as $\mathrm{Corr}$) if its released  value $y$ depends on at most other $T$ elements in $\mathcal{X}$. 
\label{def:T-aff}
\end{mydef}

Furthermore, let $\mathcal{Q}_i$ be the set of elements on which a $T$-dependent element $x_i\in\mathcal{X}$ is dependent (through model $\mathrm{Corr}_i$) ($|\mathcal{Q}_i|\leq T$), $\mathcal{Q}_i'$ be the set of released values of $\mathcal{Q}_i$, and $d_{\alpha|\mathcal{Q}_i',\mathrm{Corr}_i}$ represent the possible value(s) of $x_i$ that can be released due to the releasing of $\mathcal{Q}_i'$ and model $\mathrm{Corr}_i$. Note that it is possible for some elements to have only one  possible value to be shared under a specific correlation model. If the only possible value happens to be the true value of that element, we call  these elements  \emph{ineliminable elements}, whose privacy will be inevitably compromised for the sake of the  utility improvement of the entire shared elements (we formally investigate this issue in Section~\ref{sec:order}). As a result, we 
 propose the following definition. 

\begin{mydef}
A randomized mechanism $A$ is said to be $(\epsilon,T)$-dependent local  deferentially private for an element that is not ineliminable 
if 
\begin{equation*}
        \sup_{y\in\sigma(\mathcal{Y}),d_{\alpha|\mathcal{Q}_i',\mathrm{Corr}_i}, d_{\beta|\mathcal{Q}_i',\mathrm{Corr}_i}}\frac{\Pr(y|A(x_i=d_{\alpha|\mathcal{Q}_i',\mathrm{Corr}_i}))}{\Pr(y|A(x_i=d_{\beta|\mathcal{Q}_i',\mathrm{Corr}_i}))} \leq e^{\epsilon}.
\end{equation*}
\label{def:dldp}
\end{mydef}

Definition \ref{def:dldp} can be considered as   a specialization of the general LDP in Definition \ref{def:general-ldp} by having $d_{\alpha} = d_{\alpha|\mathcal{Q}_i',\mathrm{Corr}_i}$ and $d_{\beta} = d_{\beta|\mathcal{Q}_i',\mathrm{Corr}_i}$. Essentially, Definition \ref{def:dldp} means that any output  of a $T$-dependent element $x$ is nearly equally as likely to have come from any of its possible input values given other already shared  elements (i.e., $\mathcal{Q}_i'$ and a correlation model $\mathrm{Corr}_i$). We   summarize the main notations used throughout this paper in Table \ref{table_notations}.

\begin{table}[htb]
\small
\begin{center}
 \begin{tabular}{c | c  } 
 \hline
 \hline
 Notions & Descriptions \\
  \hline
  $x_i^j$& $SNP_i$ for individual $I_j$\\
  \hline
$d_{\alpha|\mathcal{Q}_i',\mathrm{Corr}_i}$& possible values of $x_i^j$ after SNPs in $\mathcal{Q}_i$ are \\
$\in\{0,1,2\}$& shared as $\mathcal{Q}_i'$ given the correlation model $\mathrm{Corr}_i$\\
\hline
$y_i$ & released value of SNP $x_i^j$,$y_i\in\sigma(\mathcal{Y}) = \{0,1,2\}$  \\
\hline
$\hat{\tau}$ & correlation threshold for value elimination\\
\hline
$\hat{\gamma}$ & inconsistency threshold \\
\hline
 &  a tuple describing   individual $I_j$'s MDP interaction\\ 
$\{\mathcal{S}^j,s_1^j,$ & with the environment, i.e., \{  set of all MDP states  \\
$\mathcal{A}^j,\Pr^j(\cdot),$ & of  $I_j$, initial MDP state of $I_j$, action set, \\
$\mathcal{R}^j,H^j\}$& transition probability between two MDP states, \\
&set of rewards, the horizon of the MDP \} \\
\hline
$\pi_{i}^j$ & individual $I_j$'s decision policy at time step $i$\\
\hline
${\pi_{i}^j}^*$ & individual $I_j$'s optimal decision policy at time step $i$\\
\hline
$v^{\pi_i^j}(s_i)$ &  state-value function of   state $s_i^j$ under policy $\pi_i^j$\\
\hline
$\boldsymbol{\pi}^j$ & a SNP sharing order of individual $I_j$\\
\hline
${\boldsymbol{\pi}^j}^*$ & the optimal SNP sharing order of individual $I_j$\\
     \hline
 \hline
\end{tabular}
\caption{Frequently used notations used in the paper.} 
\label{table_notations}
\end{center}
\end{table}

\subsection{Achieving $(\epsilon,T)$-dependent LDP in Genomic Data Sharing}
\label{subsec:improvePrivacy}
Our experimental results show the vulnerability of directly applying RR mechanism for genomic data sharing. Thus, here, our goal is to come up with a genomic data sharing approach achieving $(\epsilon,T)$-dependent LDP that is robust against the correlation attack. The definition of LDP states that given any output, the distinguishability between any two possible inputs needs to be bounded by $e^{\epsilon}$. In Section~\ref{subsec:attack}, all values in set $\{0,1,2\}$ are considered as possible inputs for all SNPs during data sharing. However, we know that the attacker can eliminate some input states using correlations. Hence, for the rest of the paper, we consider the possible input states as the ones  that are not eliminated by using correlations. In other words, instead of providing indistinguishability between any two values in the set $\{0,1,2\}$, we provide indistinguishability between the values that are statistically possible. 

In the correlation attack described in Section~\ref{subsec:attack}, the attacker uses two threshold values. The correlation values less than $\tau$ are considered as low correlation. In addition, if the fraction of SNPs having low correlation with a state of a particular SNP is more than $\gamma$, such state of the SNP is eliminated by the attacker. 
In the data sharing scheme, we also use these two parameters to eliminate states. However, the parameters used by the algorithm may not be same with the ones used by the attacker. Hence, to distinguish the parameters used by the algorithm and the attacker, we represent the parameters used in the algorithm as $\hat{\tau}$ and $\hat{\gamma}$ (which are the design parameters of the proposed data sharing algorithm). We describe this algorithm for a donor $I_j$ as follows.

In each step of the proposed algorithm, one SNP $x_i^j$ is processed. The algorithm first determines the states to be eliminated by considering previously shared SNPs. Then, the algorithm selects the value to be shared ($y_i^j$) by limiting the distinguishability of non-eliminated states by $e^{\epsilon}$. Hence, the order of processing may change the number of eliminated states for a SNP, which may also change the utility of the shared data. 
For instance, when a SNP is processed as the first SNP, all its three states are possible (for sharing) since there is no previously shared SNP. However, processing the same SNP as the last SNP may end up eliminating one or more of its states (due to their correlations with previously shared SNPs). We propose an algorithm to select the optimal sharing order (considering utility of shared data) in Section~\ref{sec:order}. In the following, we assume that a sharing order is provided by the algorithm in Section~\ref{sec:order} and SNPs are processed one by one following this order.

For $x_i^j$, the algorithm considers the previously shared data points and identifies the states which will be eliminated. As explained in the correlation attack, the algorithm counts the number of previously shared SNPs which have low correlation with states 0, 1, and 2 of $x_i^j$. Thus, the algorithm keeps counts for the previously shared SNPs ($SNP_k$) having $\Pr(x_i^j = 0~|~ x_k^j = y_{k}^j) < \hat{\tau}$, $\Pr(x_i^j = 1 ~|~ x_k^j = y_{k}^j) < \hat{\tau}$, $\Pr(x_i^j = 2 ~|~ x_k^j = y_{k}^j) < \hat{\tau}$ as $\hat{c_{i,0}^j}$, $\hat{c_{i,1}^j}$, and $\hat{c_{i,2}^j}$, respectively. If any of these values is greater than or equal to $\hat{\gamma} \cdot i$, the algorithm eliminates such value from the possible outputs of $x_i^j$. Let $p = e^\epsilon / (e^\epsilon + 2)$ and $q = 1 / (e^\epsilon + 2)$, and the value of $x_{i}^j$ be $0$. Then, the algorithm assigns the probabilities of non-eliminated states as follows:
\begin{itemize}
    \item If there are three possible outputs (i.e., no eliminated state), the algorithm uses the same probability distribution with the RR mechanism as $(p, q, q)$. Thus, $\Pr(y_{i}^j=0) = p$ and $\Pr(y_{i}^j = 1) = \Pr(y_{i}^j = 2) = q$.
    \item If there are two possible outputs (i.e., one eliminated state) and $x_i^j$ (state $0$) is not eliminated, the algorithm uses an adjusted probability distribution as $(p / (p + q), q / (p + q), 0)$ (or $(p / (p + q), 0, q / (p + q))$, depending on which state is eliminated).
     \item If there are two possible outputs (i.e., one eliminated state) and $x_i^j$ is eliminated, the algorithm uses an adjusted probability distribution as $(0, 0.5, 0.5)$.
    \item If there is one possible output (i.e., two eliminated states), the corresponding state is selected as the output.
    \item If there is no possible output (i.e., three eliminated states), the algorithm uses the same probability distribution as the RR mechanism.
\end{itemize}
For other values of $x_{i}^j$, the algorithm also works in a similar way. 
The probability distributions for sharing a data point are also shown in Figure~\ref{fig:prob1}. Based on these probabilities, the algorithm selects the value of $y_{i}^j$. 
If the attacker knows $\hat{\tau}$ and $\hat{\gamma}$ used in the algorithm, it can compute the possible values for each SNP using perturbed data $\mathcal{Y}^j = \lbrack y_{1}^j, \ldots,y_{l}^j \rbrack$, $\hat{\tau}$, $\hat{\gamma}$ and the correlations between the SNPs. Since $e^{\epsilon}$ ratio is preserved in each case, the attacker can only distinguish the possible inputs with $e^{\epsilon}$ difference.
\begin{figure}
\centering
\includegraphics[width=10cm,keepaspectratio]{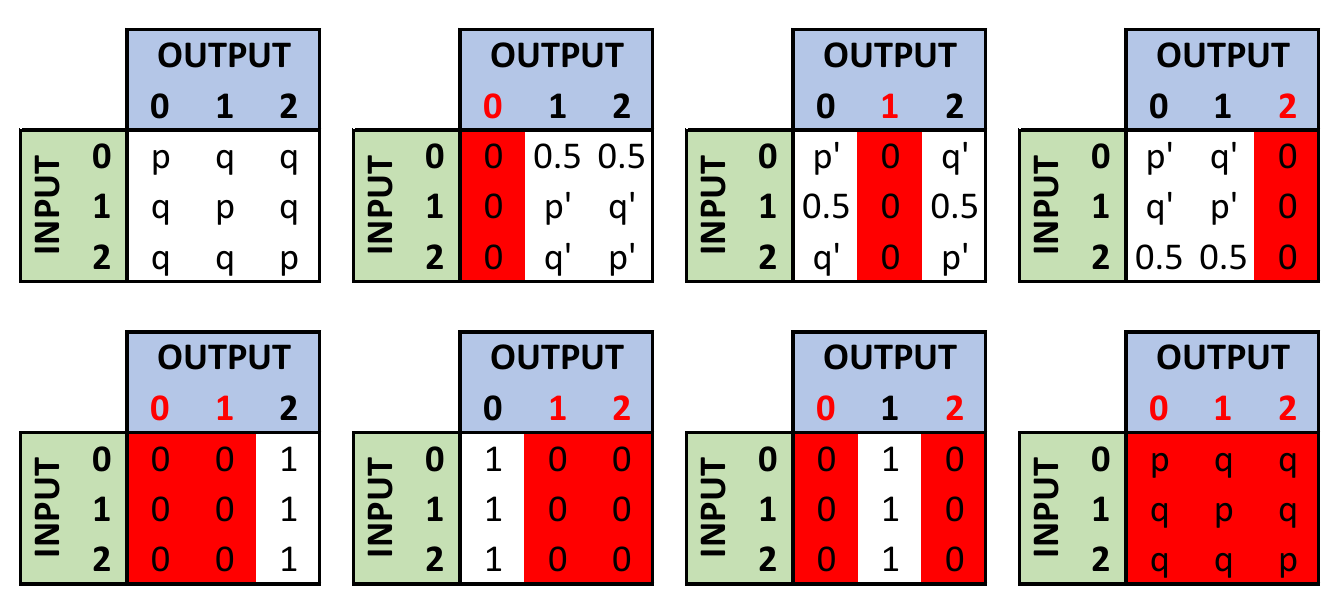}
\caption{Probability distribution used by the data sharing mechanism after eliminating states using correlations. $p = e^\epsilon / (e^\epsilon + 2)$, $q = 1 / (e^\epsilon + 2)$, $p' = p / (p + q)$, and $q' = q / (p + q)$. Red columns represent the eliminated states.}
\label{fig:prob1}
\end{figure}

\subsection{Improving Utility by Adjusting Probability Distributions}
\label{subsec:improveUtility}
In Section~\ref{subsec:improvePrivacy}, we proposed a data sharing mechanism to improve the privacy of randomized response mechanism against the correlation attack. The mechanism guarantees that the perturbed data $\mathcal{Y}^j = \lbrack y_{1}^j, \ldots,y_{l}^j \rbrack$ belonging to $I_j$ does not include any value that have low correlation with other SNPs. However, consistent with existing LDP-based mechanisms, the algorithm assigns equal sharing probabilities for each incorrect value of a SNP $i$. That is, probability of sharing any non-eliminated value $y_{i}^j$ (such that $y_{i}^j \neq x_{i}^j$) is the same. This also implies that utility of each incorrect value of a SNP are the same (for the data collector). However, this may cause significant utility loss since the accuracy of genomic analysis may significantly decrease as the values of shared SNPs deviate more from their original values (e.g., in genomic data sharing beacons or when studying haploinsufficiency). For genomic data, typically, the utility of each value of a SNP is different and the utility of a SNP may change depending on the purpose of data collection. Here, our goal is to improve the utility of shared data by modifying the probability distributions without violating $(\epsilon,T)$-dependent LDP.

To improve utility, we introduce new probability distributions, such that, for each shared SNP, the probability of deviating high from its ``useful values'' is small. Useful values of a SNP depend on how the data collector intends to use the collected SNPs. For instance, for genomic data sharing beacons, changing the value of a shared SNP with value $2$ to $1$ does not decrease the utility, but sharing it as $0$ may cause a significant utility loss. 
Similarly, while studying haploinsufficiency, obfuscating a SNP with value $2$ results in a significant utility loss while changing a $0$ to $1$ (or $1$ to $0$) does not cause a high utility loss. 
Here, to show how the proposed data sharing mechanism improves the utility, we focus on genomic data sharing beacons (similar analysis can be done for other uses of genomic data as well). 

\begin{figure}
\centering
\includegraphics[width=10cm,keepaspectratio]{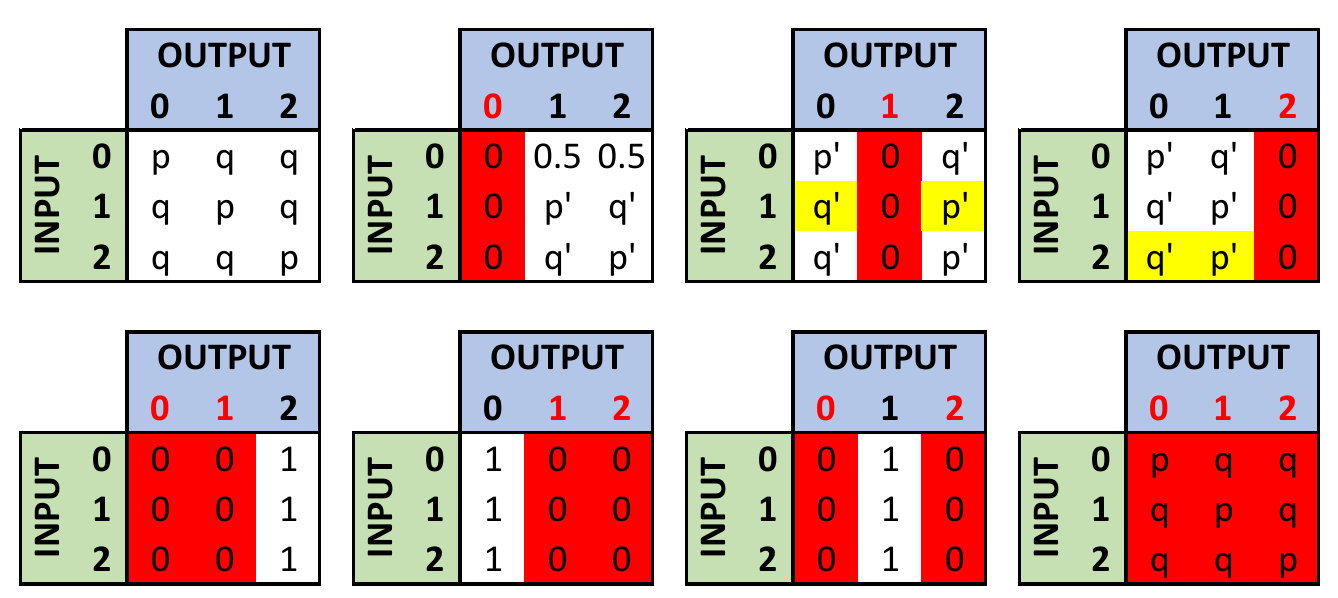}
\caption{Probability distribution used by the data sharing mechanism to improve utility of beacon queries (in genomic data sharing beacons). $p = e^\epsilon / (e^\epsilon + 2)$, $q = 1 / (e^\epsilon + 2)$, $p' = p / (p + q)$, and $q' = q / (p + q)$. Red columns represent the eliminated states. The differences with Figure~\ref{fig:prob1} are highlighted with yellow.}
\label{fig:prob2}
\end{figure}

Genomic data sharing beacons allow users (researchers) to learn whether individuals with specific alleles (nucleotides) of interest are present in their dataset. A user can submit a query, asking whether a genome exists in the beacon with a certain nucleotide at a certain position, and the beacon answers as ``yes'' or ``no''. Since having at least one minor allele is enough for a ``yes'' answer, having one minor allele (a SNP value of $1$) or two minor alleles (a SNP value of $2$) at a certain position is equivalent in terms of the utility of beacon's response. Therefore, if the correct value of a SNP is $1$ or $2$, sharing the incorrect value as $2$ or $1$ will have higher utility than sharing it as $0$. Considering this, we change the probability distributions of the data sharing mechanism (given in Section~\ref{subsec:improvePrivacy}) as follows to improve the utility. 
\begin{itemize}
    \item If there are three, one, or no possible outputs, the probability distributions described in Section~\ref{subsec:improvePrivacy} are used.
    \item If there are two possible outputs and $x_i^j$ is not eliminated, the algorithm shares $y_{i}^j$ as $x_{i}^j$ with probability $p / (p + q)$ and the incorrect value with probability $q / (p + q)$.
    \item If there are two possible outputs, $x_i^j$ is eliminated, and $x_i^j = 0$, the algorithm uses an adjusted probability distribution as $(0, 0.5, 0.5)$.
    \item If there are two possible outputs, $x_i^j$ is eliminated, and $x_i^j = 1$, the algorithm uses an adjusted probability distribution as $(q / (p + q), 0, p / (p + q))$.
    \item If there are two possible outputs, $x_i^j$ is eliminated, and $x_i^j = 2$, the algorithm uses an adjusted probability distribution as $(q / (p + q), p / (p + q), 0)$.
\end{itemize}

As in Section \ref{subsec:improvePrivacy}, $p = e^\epsilon / (e^\epsilon + 2)$ and $q = 1 / (e^\epsilon + 2)$. These probability distributions (which are also shown in Figure~\ref{fig:prob2}) still preserve the $e^{\epsilon}$ ratio between states. Note that for eliminating the states, the same process is used as described in Section~\ref{subsec:improvePrivacy}. To determine the processing order of the shared SNPs, the algorithm in Section~\ref{sec:order} is used.

\subsection{Proposed Data Sharing Algorithm}
\label{subsec:algorithm}

In Section~\ref{subsec:improvePrivacy}, we described how to improve privacy by eliminating statistically unlikely values for each SNP. In Section~\ref{subsec:improveUtility}, we explained how to modify probability distributions to improve utility of shared data for genomic data sharing beacons. Using these two ideas, we describe our proposed genomic data sharing algorithm in the following and provide the details for an individual $I_j$ in Algorithm~\ref{alg:alg1}. The algorithm processes all SNPs one by one and in each iteration, it computes a value to share for the SNP being processed (eventually, all SNPs are processed and they are shared at the same time with the data collector). The algorithm first eliminates the states having low correlations with the previously processed SNPs, as described in Section~\ref{subsec:improvePrivacy}. Two thresholds $\hat{\tau}$ and $\hat{\gamma}$ are used to determine the eliminated states. We evaluate the effect of these threshold values on utility and privacy in Section~\ref{subsec:parameters}. Then, the algorithm decides the shared value of the SNP using the probability distribution in Figure~\ref{fig:prob2}. This process is repeated for all SNPs and the SNP sequence to be shared (i.e., output) is determined. Since we consider all pairwise correlations, changing the order may change the utility of the proposed scheme by eliminating different states. We discuss the optimal selection of this order (in terms of utility) in Section~\ref{sec:order} and Algorithm~\ref{Value_Iteration} outputs the optimal order for each individual $I_j$ (i.e.,  ${\boldsymbol{\pi}^j}^*$). Due to the computational complexity of Algorithm~\ref{Value_Iteration}, we also propose a greedy algorithm in Section \ref{sec:heuristic}. Thus, either the output of the optimal or the greedy algorithm is used as the input for the proposed data sharing algorithm. In addition, we also visualize the selection of next SNP to be processed via a tree structured flowchart in Figure~\ref{fig:treei} (in Appendix \ref{app:tree}). Furthermore, in genomic data sharing, there is an interdependent privacy issue between genome donors and their family members. Thus, we also discuss the effect of the proposed data sharing mechanism on the privacy of donors' family members in Section~\ref{sec:discussion} and propose an algorithm to determine the maximum privacy budget (in terms of the $\epsilon$ parameter) that can be used by a genome donor while considering the privacy budgets of her family members.

\begin{algorithm}
\footnotesize
\SetKwInOut{Input}{input}
\SetKwInOut{Output}{output}
\Input{Original data $\mathcal{X}^j = \lbrack x_{1}^j, \ldots, x_{l}^j \rbrack$ of $I_j$,  processing order $\boldsymbol{\pi}^j$, privacy budget $\epsilon$, correlation threshold $\hat{\tau}$, inconsistency threshold $\hat{\gamma}$, pairwise correlations between data points.}
\Output{Perturbed data $\mathcal{Y}^j = \lbrack y_{1}^j, \ldots, y_{l}^j \rbrack$.}
      \ForAll{$a \in \{1,2,\ldots,l\}$}{
        $i \longleftarrow \boldsymbol{\pi}^j(a)$;
        
        $\hat{c_{i,0}^j}$, $\hat{c_{i,1}^j}$, $\hat{c_{i,2}^j}$ $\longleftarrow 0$;
        
        \ForAll{$b \in  \{1,2,\ldots,a-1\}$}{
            $k \longleftarrow \boldsymbol{\pi}^j(b)$;
            
            \uIf{$\Pr(x_i^j = 0~|~ x_k^j = y_{k}^j) < \hat{\tau}$} { $\hat{c_{i,0}^j} \longleftarrow \hat{c_{i,0}^j} + 1$;
	     	}
	     	\uElseIf{$\Pr(x_i^j = 1~|~ x_k^j = y_{k}^j) < \hat{\tau}$} { 
	     	$\hat{c_{i,1}^j} \longleftarrow \hat{c_{i,1}^j} + 1$;
	     	}
	     	\uElseIf{$\Pr(x_i^j = 0~|~ x_k^j = y_{k}^j) < \hat{\tau}$}{
	     $\hat{c_{i,2}^j} \longleftarrow \hat{c_{i,2}^j} + 1$;
	     	}
        }
        \uIf{$\hat{c_{i,0}^j} \geq \hat{\gamma} \cdot a$} { eliminate state 0;
	     		}\uElseIf{$\hat{c_{i,1}^j} \geq \hat{\gamma} \cdot a$} { eliminate state 1;
	     		}\uElseIf{$\hat{c_{i,2}^j} \geq \hat{\gamma} \cdot a$} { eliminate state 2;
	     		}
        
        $y_i^j \longleftarrow$ random value from non-eliminated states using probability distribution in Figure \ref{fig:prob2}.
     }
		\caption{\color{black} Genomic data sharing scheme for donor $I_j$ under $(\epsilon,T)$-dependent LDP. \color{black}
		}
\label{alg:alg1}
\end{algorithm}

\begin{lemma}
Given  a processing order, Algorithm \ref{alg:alg1} achieves $(\epsilon,l-1)$-dependent local differential privacy for each genomic data point that is not ineliminable. 
\label{privacy_proof}
\end{lemma}

\begin{proof}[Proof]
The proof directly follows from the reallocation of probability mass used in the RR mechanism. Since $\mathrm{Corr}_i$ is the pairwise correlation between SNPs, we have $T=l-1$. Besides, the $e^{\epsilon}$ ratio is preserved in the modified RR mechanism, and hence the condition in Definition \ref{def:dldp} can always hold for ineliminable SNPs.
\end{proof}
\section{Optimal Data Processing Order for the Proposed Genomic Data Sharing Mechanism}
\label{sec:order}

Algorithm~\ref{alg:alg1} considers/processes one SNP at a time and as discussed, different processing orders may cause elimination of different states of a SNP, which, in turn, may change the utility of the shared data. Assuming there are totally $l$ SNPs in $\mathcal{X}^j$ of an individual $I_j$, Algorithm~\ref{alg:alg1} can process them in $l!$ different orders.  
As a result, determining an optimal order of processing to maximize the utility of the shared sequence of SNPs is a critical and challenging problem. In this section, we formulate the problem of determining the optimal order of processing as a Markov Decision Processes (MDP)~\cite{sutton2018reinforcement}, which can be solved by value iteration using dynamic programming. Note that the algorithm locally processes all SNPs, and then perturbed data is shared all at once. Hence, the data collector does not see or observe the order of processing. 

Since we consider genomic data sharing beacons to study the utility of shared data (as in Section~\ref{subsec:improveUtility}) and the proposed sharing scheme is non-deterministic, we aim at achieving the maximum expected utility for the beacon responses using the shared SNPs. Note that similar analysis can be done for other uses of genomic data as well. Beacon utility is typically measured over a population of individuals, however, in this work, we consider an optimal processing order, which maximizes the expected beacon utility for each individual. The reason is twofold: (i) an individual does not have access to other individuals' SNPs and (ii) a population's maximum expected beacon utility can be achieved if all individuals' maximum expected beacon utility are obtained due to the following Lemma.

\begin{lemma}
Maximizing the expectation of individuals' beacon utility is a sufficient condition for maximizing the expectation of a population's beacon utility.
\label{individual_utility}
\end{lemma}
\begin{proof}[Proof]
Proof is given in Appendix \ref{app:proof}
\end{proof}

The sufficient condition in Lemma~\ref{individual_utility} can easily be extended to other genomic data sharing scenarios as long as the individuals share their SNPs independently from each other.

\subsection{Determining the Optimal Processing Order via Markov Decision Processes (MDP)}
Here, we proceed with obtaining the optimal order of processing which maximizes individuals' expected beacon utility. First, we model the SNP state elimination and   processing order as an agent-environment interaction framework, where the agent is a specific individual (donor), the environment is the proposed SNP sharing scheme considering correlations (in Algorithm~\ref{alg:alg1}), and the interaction follows a MDP.

For instance, consider the individual (donor) $I_j$ in the population. Then, her MDP interaction with the environment is characterized as a tuple $\{\mathcal{S}^j,s_1^j,\mathcal{A}^j,\Pr^j(\cdot),\mathcal{R}^j,H^j\}$, where $\mathcal{S}^j$ is the set of all MDP states of individual $I_j$, $s_1^j$ is her initial MDP state, $\mathcal{A}^j$ is her action set, $\Pr^j(\cdot)$ is the transition probability between two MDP states of $I_j$, $\mathcal{R}^j$ is her set of rewards, and $H^j$ is the horizon of the MDP (i.e., number of rounds in discrete time). In our case, $H=l$ (number of SNPs to be processed), and  $s_1^j=\emptyset$. 
At each time step $i\in\{1,2,\cdots,H\}$ (i.e., when individual $I_j$ processes her $i$th SNP), the agent chooses an action $a_i^j$ from her action pool $\mathcal{A}_i^j\subset \mathcal{A}^j$ (i.e.,  selects a specific SNP from her remaining unprocessed SNPs), where $\mathcal{A}_i^j$ is the set of remaining unprocessed SNPs and $\mathcal{A}^j$ is the set of all SNPs of individual $I_j$. 
Then, the environment provides the agent with a MDP state $s_i^j$ and a reward $r_i^j$. In particular, $s_i^j = \{y_1^j,y_2^j,\cdots,y_i^j\}$ (i.e., the list  recording all  observations of previously processed SNPs of individual $I_j$) and $r_i^j$ is the utility of the beacon response on $y_i^j$, and hence, we have $r_i^j\in\mathcal{R}^j=\{0,1\}$. 
After observing $s_i^j$ and receiving $r_i^j$, the agent takes the next action $a_{i+1}^j\in\mathcal{A}_{i+1}^j$, which causes $s_i^j$ to transit to $s_{i+1}^j$ via the transition probability $\Pr(s_{i+1}^j|s_i^j,a_i^j,s_{i-1}^j,a_{i-1}^j,\cdots,s_1^j,a_1^j) = \Pr(s_{i+1}^j|s_i^j,a_i^j)$. Here, the equality holds due to the Markov property~\cite{sutton2018reinforcement} and $\Pr(s_{i+1}^j|s_i^j,a_i^j)$ is determined by the leaf nodes in Figure~\ref{fig:treei} (in Appendix \ref{app:tree}). An illustration of the  MDP interaction between the agent (individual $I_j$) and the environment (Algorithm \ref{alg:alg1}) at time step $i$ (processing the $i$th SNP) is shown in Figure~\ref{fig:DP}.
\begin{figure}
    \centering
    \includegraphics[width=0.45\textwidth]{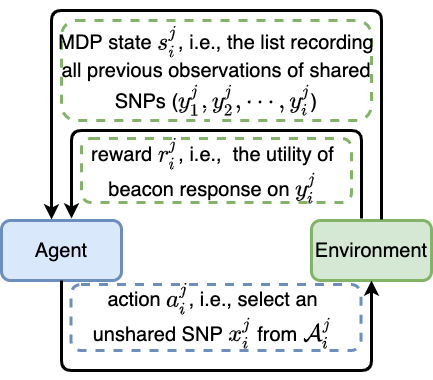}
    \caption{The interaction between the agent (individual $I_j$) and the environment (Algorithm \ref{alg:alg1}) at time step $i$, i.e., when processing the $i$th SNP ($x_i^j$) of individual $I_j$.} 
    \label{fig:DP}
\end{figure}

Since the optimal order can be predetermined and should be invariant in time, we model the agent's (individual $I_j$) decision policy at time step $i$ as a deterministic mapping as $\pi_i^j: \mathcal{S}^j\rightarrow\mathcal{A}_i^j$, i.e., $a_{i}^j = \pi_i^j(s_i)$, $\forall i \in\{1,2,\cdots,l\}$. Let $\boldsymbol{\pi}^j = \{\pi_1^j,\pi_2^j,\cdots,\pi_l^j\}$ be the sequence of decision policies of the agent. Due to the non-deterministic behavior of Algorithm~\ref{alg:alg1}, we characterize the environment's behavior on individual $I_j$ as a probabilistic mapping as $p:\mathcal{S}^j\times\mathcal{A}_i^j\rightarrow\mathcal{S}^j\times\mathcal{R}^j$ (i.e., $\Pr(s_{i+1}^j|s_i^j,a_i^j)$). Furthermore, we define the future cumulative return for individual $I_j$ starting from MDP state $s_i^j$ as $R_i^j=\sum_{\eta=i}^{\eta=l}r_{\eta}^j$ and the state-value function of MDP state $s_i^j$ under policy $\pi_i^j$ as $v^{\pi_i^j}(s_i) = \mathbb{U}_p[R_{i+1}^j|s_i^j]$ ($\mathbb{U}_{p}[\cdot]$ indicates that utility is considered in an expected manner with respect to the environment's probabilistic mapping $p$). Then, to maximize an individual's expected beacon utility at time step $i$, the agent takes the optimal decision ${\pi_i^j}^*\in \text{argmax}_{\pi_i^j}v^{\pi_i^j}(s_i), \forall s_i^j\in\mathcal{S}^j$ and $\in$ suggests that ${\pi_i^j}^*$ may not be unique.

Thus, we have formulated the optimal order of processing problem as a finite-horizon MDP problem, whose state, action, and reward sets are all finite and dynamics are characterized by a finite set of probabilities (i.e., $\Pr(s_i^j|s_{i-1}^j,a_{i-1}^j)$). The finite-horizon MDP problem is P-complete, as it can be reduced from the circuit value problem, which is a well-known P-complete problem~\cite{papadimitriou1987complexity}. In the literature, exact optimal solution of finite-horizon MDP problem can be obtained by quite a few methods, for example value iteration, policy iteration, or linear programming~\cite{bertsekas1995dynamic}. In Algorithm~\ref{Value_Iteration}, we provide a value iteration~\cite{sutton2018reinforcement} based approach to determine the optimal order of processing for an individual. 

Algorithm~\ref{Value_Iteration} is implemented using dynamic programming starting from the last time step, and it has a computational complexity of  $\mathcal{O}(|\mathcal{S}^j|^2|\mathcal{A}^j|)$ for individual $I_j$~\cite{sutton2018reinforcement}. For finite-horizon MDP, the number of MDP states grows exponentially with the number of variables, which is known as the curse of dimensionality. For example, in our case, at time step $i$, Algorithm~\ref{Value_Iteration} needs to calculate the state-value function for $3^i$ states. In the literature, many approaches have been proposed to address this issue, such as state reduction~\cite{dean1997model} and logical representations~\cite{boutilier2000stochastic}, which, however are outside the scope of this paper. Therefore, Algorithm~\ref{Value_Iteration} may be computationally expensive to process large amount of data, and hence in the following section, we propose a heuristic approach to process long sequence of SNPs.

\begin{algorithm}
\footnotesize
\SetKwInOut{Input}{input}
\SetKwInOut{Output}{output}
\Input{the MDP tuple $\{\mathcal{S}^j,s_1^j,\mathcal{A}^j,\Pr^j(\cdot),\mathcal{R}^j,H^j\}$ of $I_j$.}
\Output{the optimal order of processing that maximizes  individual $I_j$'s expected beacon utility, i.e., ${\boldsymbol{\pi}^j}^* = \{{\pi^j}^*(s_1^j),{\pi^j}^*(s_2^j),\cdots,{\pi^j}^*(s_{l}^j)\}$.}
\ForAll{$i\in\{1,2,\cdots,l\}$}{
randomly initialize $v(s_i^j)$, $\forall s_i^j\in\mathcal{S}^j$;

randomly initialize a positive parameter $\delta$;

\While{$\delta> 0$}{
      \ForAll{\texttt{$s_i^j\in\mathcal{S}^j$}}{
        $c\leftarrow v(s_i^j)$;
        
        $v(s_i^j)\leftarrow \text{max}_{a^j}  \sum p(s_{i+1}^j,r_{i+1}^j|s_i^j,a^j)[r_i^j+v(s_{i+1}^j)]$; \label{probx}
        
        $\delta \leftarrow \text{max}\{\delta,|c-v(s_{i}^j)|\}$;
     }
     }
         ${\pi^j}^*(s_i^j) = \text{argmax}_{a^j}p(s_{i+1}^j,r_{i+1}^j|s_i^j,a^j)[r_i^j+v(s_{i+1}^j)]$;
     }		\caption{Determining the optimal order of processing for individual $I_j$.}
\label{Value_Iteration}
\end{algorithm}

\subsection{A Heuristic Approach}\label{sec:heuristic}

In this work, we consider sharing  thousands of SNPs of individuals in a  population. As a consequence, it is computationally prohibitive to obtain the exact optimal order of processing for each individual. 
We propose the following  heuristic approach for an individual $I_j$ to process her SNPs in a local greedy manner. Specifically, at each time step $i$, the algorithm selects the SNP with the maximum expected beacon utility, i.e., $a_i^j = \text{argmax}_{a^j\in\mathcal{A}_i^j}\mathbb{U}_{a^j}$, where $\mathcal{A}_i^j$ is the set of remaining SNPs of individual $I_j$, and $\mathbb{U}_{a^j}$ denotes the expected immediate utility if individual $I_j$ selects SNP $a^j$ and it is determined by a certain leaf node in Figure~\ref{fig:treei} (in Appendix \ref{app:tree}). After evaluating the condition in the root node using a SNP, only one leaf node can be activated. For example, without loss of generality, assume that at time step $l-1$, SNPs $x_i^j$ and $x_k^j$ are left in $\mathcal{A}_{l-1}^j$, after the elimination check, $x_i^j$ can activate the leaf node characterized by $(0,\frac{p}{p+q},\frac{q}{p+q})$ and $U=1$, and $x_k^j$ can activate the leaf node characterized by $(\frac{q}{p+q},0,\frac{q}{p+q})$ and $U=\frac{q}{p+q}$. Then, the heuristic algorithm selects SNP $x_i^j$ to process at time step $l-1$. If there is a tie between two SNPs, we randomly choose one. We compare this heuristic approach with the optimal algorithm (in Algorithm~\ref{Value_Iteration}) in Section~\ref{subsec:orderEval}.
\section{Evaluation}
\label{sec:evaluation}

We implemented the proposed data sharing scheme in Section~\ref{subsec:algorithm} and used a real genomic dataset containing the genomes of the Utah residents with Northern and Western European ancestry (CEU) population of the HapMap project~\cite{international2003international} for evaluation. We used 1000 SNPs of 156 individuals from this dataset for our evaluations. Using this dataset, we computed all pairwise correlations between SNPs. For each 1 million ($1000\times 1000$) SNP pairs, we computed 9 ($3\times3$) conditional probabilities. Hence, we totally computed 9 million conditional probabilities (for all pairwise correlations between all SNPs). Note that, to quantify the privacy of the proposed scheme against the strongest attacks, we used the same dataset to compute the attacker's background knowledge. However, in practice, the attacker may use different datasets to compute such correlations and its attacks may become less successful when less accurate statistics are used. We also assumed that each donor has the same privacy budget ($\epsilon$). To quantify privacy, we used the attacker's estimation error. Estimation error is a commonly used metric to quantify genomic privacy~\cite{wagner2017evaluating}, which quantifies the average distance of the attacker's inferred SNP values from the original data ($\mathcal{X}^j$) as 
\begin{displaymath}
E = \frac{\sum\limits_{v\in\{0,1,2\}; k\in\{1,...,l\}} \Pr(x_k^j = v)||x_k^j-v||}{l},
\end{displaymath}
where $\Pr(x_k^j = v)$ is the attacker's inference probability for $x_k^j$ being $v$. We assume the attacker's only knowledge is $p$ and $q$ initially, which are computed based on $\epsilon$. Then, using the correlations, the attacker improves its knowledge by eliminating the statistically less likely values. For the eliminated states, attacker sets the corresponding probability to $0$. Since $||x_k^j-v||$ can be at most 2 for genomic data, $E$ is always in the range $[0,2]$, where higher $E$ indicates better privacy. Thus, when the attacker's estimation error decreases, the inference power of the attacker (e.g., to infer the predisposition of a target individual to a disease) increases accordingly. To quantify the utility, we used the accuracy of beacon responses. For each SNP, we first run the beacon queries using the original values and then run the same queries with the perturbed values. Let the number of beacon responses (SNPs) for which we obtain the same answer for both original data and perturbed data be $n_s$. We computed the accuracy as $A = n_s / l$ ($l$ is the total number of beacon queries), which is always in the range $[0,1]$.

In the following, we first compare the proposed algorithm with the original RR mechanism in terms of privacy and utility. Then, we evaluate the effect of the design parameters on privacy and utility. Finally, we show the effect of the order of processing on utility.

\subsection{Comparison with the Original Randomized Response Mechanism}
\label{subsec:compareLDP}

As we discussed in Section~\ref{subsec:attack}, the original randomized response (RR) mechanism is vulnerable to correlation attacks because when a given state of a SNP is loosely correlated with at least $\gamma \cdot l$ other SNPs, the attacker can eliminate that state, and hence improve its inference power for the correct value of the SNP. In Figure~\ref{fig:ldpComp}, we show this vulnerability in terms of attacker's estimation error (blue and red curves in the figure). We observed that attacker's estimation error is the smallest (i.e., its inference power is the strongest) when the correlation threshold of the attacker ($\tau$) is $0.02$ and inconsistency threshold of the attacker ($\gamma$) is $0.03$, and hence we used these parameters for the attack. 

Under the same settings, we also computed the estimation error provided by the proposed algorithm when $\hat{\tau} = 0.02$ and $\hat{\gamma} = 0.03$. Therefore, during data sharing, we eliminated states of the SNPs having correlation less than $\hat{\tau}  = 0.02$ (the correlation threshold of the algorithm) with at least $\hat{\gamma} = 0.03$ of the previously shared SNPs (in Section~\ref{subsec:parameters}, we also evaluate the effect of these parameters on privacy and utility). 
We also let the attacker conduct the same attack in Section~\ref{subsec:attack} with the same attack parameters as before. Figure~\ref{fig:ldpComp} shows the comparison of the proposed scheme with original RR mechanism (green curve in the figure is the privacy provided by the proposed scheme). The results clearly show that the proposed method improves the privacy provided by RR after correlation attack. For instance, for $\epsilon = 1$, the proposed scheme provides approximately $25\%$ improvement in privacy compared to the RR mechanism. Note that the privacy of RR before the attack (blue curve in the figure) is computed by assuming the attacker does not use correlations. Hence, when the attacker uses correlations, it is not possible to reach that level of privacy with any data sharing mechanism and the privacy inevitably decreases. With the proposed scheme, we reduce this decrease in the privacy. To observe the limits of the proposed approach, we performed the correlation attack by assuming the attacker has 0 value for all SNPs (which is the mostly observed value in genomic data) and we observed the attacker's estimation error as 0.66 (under the same experimental settings) after the correlation attack. Hence, with any mechanism it is not possible to exceed 0.66 after correlation attack and the privacy provided by the proposed scheme is remarkable.

\begin{figure}
\centering
\includegraphics[width=8cm,keepaspectratio]{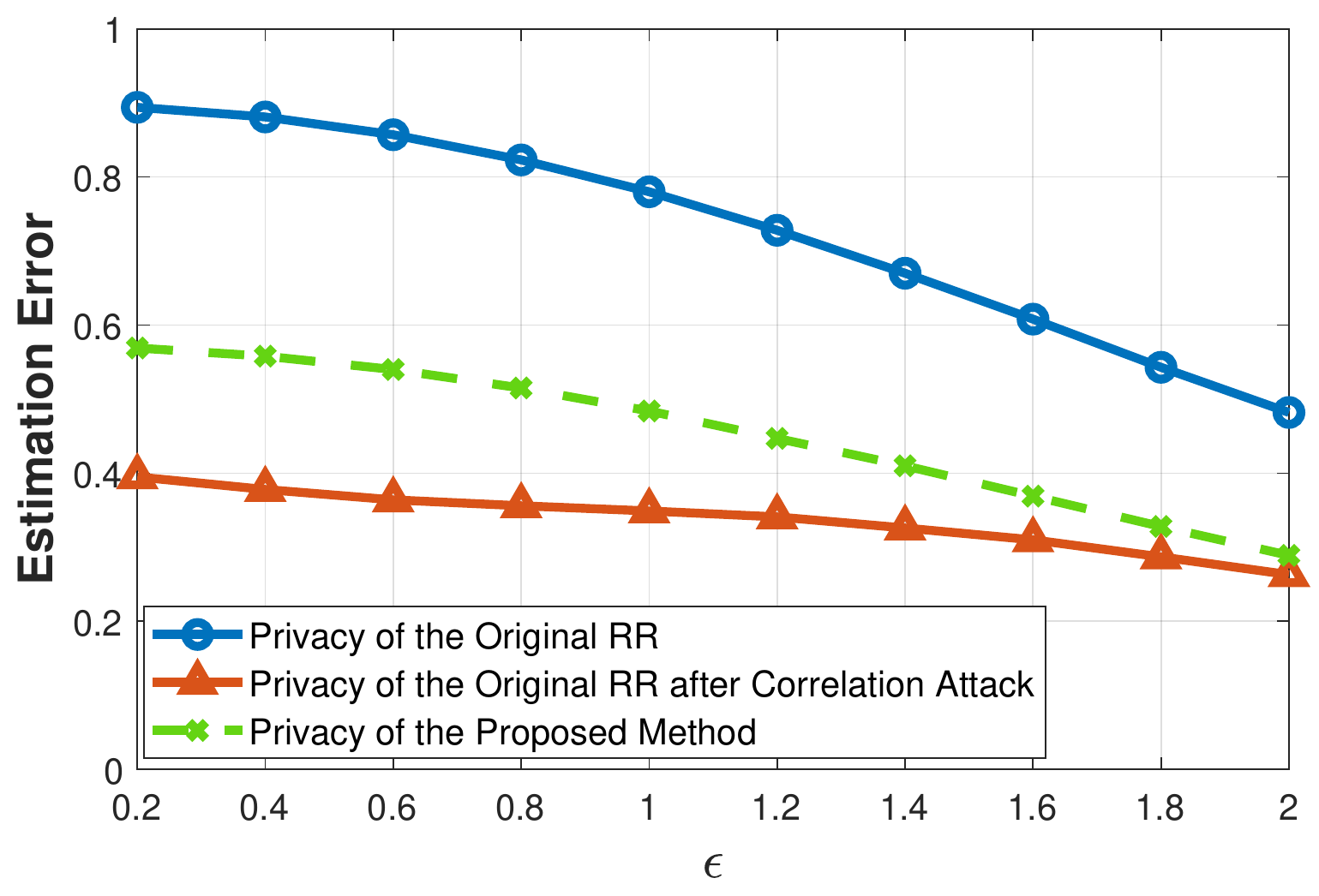}
\caption{Comparison of the proposed data sharing mechanism with original RR mechanism in terms of attacker's estimation error.}
\label{fig:ldpComp}
\end{figure}

Focusing on genomic data sharing beacons, we also compared the utility of shared data using the proposed scheme with the original RR mechanism in terms of accuracy of beacon answers (using the accuracy metric introduced before). We randomly selected 60 people from the population and used their 1000 SNPs to respond to the beacon queries. For 257 SNPs there was no minor allele, and hence the original response of the beacon query was ``no''. There was at least one minor allele in 60 people for the remaining 743 SNPs (and hence, the original response of the beacon query was ``yes''). 

For the original RR mechanism, we shared 1000 SNPs of 60 individuals after perturbation. In the RR mechanism, the data collector eliminates the noise by estimating the frequency of each value using the sharing probabilities as described in Section~\ref{subsec:randomized_response}. Hence, if $60 \cdot p$ or more individuals report $0$ for the value of a SNP (after perturbation), we considered the answer of beacon as ``no''. For the proposed data sharing scheme, we did not apply such an estimation since in the proposed scheme, the sharing probabilities of the states are different for each SNP.
Figure~\ref{fig:beaconYesNo} shows the accuracy of the beacon for 1000 queries. We observed that our proposed scheme provides approximately 95\% accuracy even for small values of $\epsilon$, while the accuracy of the RR mechanism is less than 70\% for small $\epsilon$ values and it only reaches to 85\% when $\epsilon$ increases. We provide the accuracy evaluation for the ``yes'' and ``no'' responses separately in Appendix~\ref{app:accuracy}. \color{black} Note that we do not quantify the utility over the probability of correctly reporting a point. We quantify the utility over the accuracy of beacon answers. When the answer of the beacon query is ``yes'', the original response of the beacon is mostly preserved after perturbation in both the original RR and the proposed mechanism, as shown in Appendix~\ref{app:accuracy} (while the proposed mechanism still outperforms the RR mechanism, especially for smaller $\epsilon$ values). On the other hand, when the original answer of a beacon query is ``no'', all individuals must report 0 for that SNP (to preserve the accuracy of the response). In this case, applying the original RR cannot provide high accuracy when $\epsilon$ is small, because with high probability, at least one individual reports its SNP value as 1 or 2 (i.e., incorrectly). As we also show in Appendix~\ref{app:accuracy}, our proposed approach significantly outperforms the RR mechanism in terms of the accuracy of the ``no'' responses. \color{black} Therefore, we conclude that the proposed scheme provides significantly better utility than the original RR mechanism. 

Although here we evaluated utility for genomic data sharing beacons, similar utility analyses can be done for other applications as well. Since the proposed scheme eliminates statistically unlikely values, the proposed scheme will still outperform the original RR mechanism under similar settings. Since the proposed data sharing mechanism considers the correlations with the previously shared data points (as in Algorithm~\ref{alg:alg1}) its computational complexity is $\mathcal{O}(l^2)$, where $l$ is the number of shared SNPs of a donor.

One alternative approach to improve privacy in the original RR mechanism can be adding a post-processing step that includes identifying the SNPs having low correlations with the other SNPs and replacing them with the values that have high correlations. Such an approach can be useful to prevent correlation attacks due to eliminating less likely values. However, this approach provides much lower utility compared to the proposed mechanism since the proposed mechanism improves utility by adjusting probability distributions and optimizing the order of processing. We also implemented this alternative post-processing approach and compared with the proposed mechanism. We observed similar estimation error with the proposed mechanism, which shows that this approach can also prevent correlation attacks. However, as shown in Table \ref{table:postprocess} in Appendix \ref{app:postprocess}, post-processing approach provides even lower utility than the original RR mechanism without post-processing, because it becomes harder to do efficient estimation after the post-processing. Hence, the proposed mechanism outperforms the original RR mechanism even if post-processing is applied. 

\begin{figure}
\centering
\includegraphics[width=8cm,keepaspectratio]{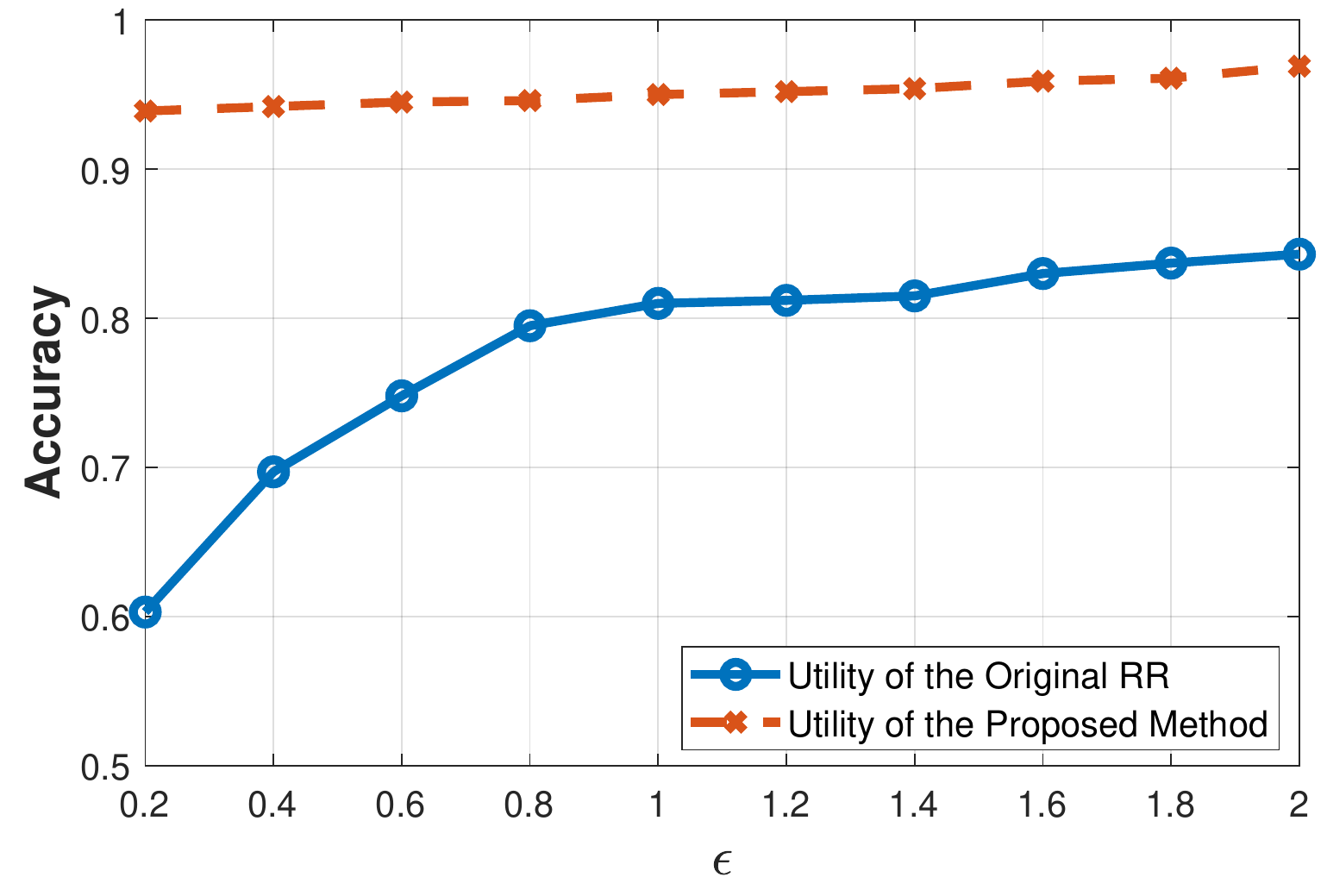}
\caption{Comparison of the proposed method with original randomized response mechanism in terms of utility. Utility is measured as the accuracy of responses provided from a genomic data sharing beacon.}
\label{fig:beaconYesNo}
\end{figure}

\subsection{The Effect of Parameters on Utility and Privacy}
\label{subsec:parameters}

\begin{table}
\small
\centering
\caption{Privacy (in terms of estimation error) of the original randomized response mechanism after the correlation attack for different values of $\gamma$ (inconsistency threshold of the attacker) when $\tau = 0.02$ and $\epsilon = 1$.}  
\begin{tabular}{ c|c|c|c|c|c }
\hline
$\gamma$ & 0.01 & 0.02 & 0.03 & 0.04 & 0.05 \\
\hline
Estimation error ($E$) & 0.491 & 0.380 & 0.348 & 0.368 & 0.415 \\
\hline
\end{tabular}
\label{table:gamma}
\end{table}

In Section~\ref{subsec:compareLDP}, we used the correlation threshold of the attacker ($\tau$) as 0.02 and inconsistency threshold of the attacker ($\gamma$) as 0.03 in its correlation attack. In our experiments, these parameters provided the strongest attack against the original RR mechanism. In Table~\ref{table:gamma}, we show how the estimation error of the attacker changes for different values of $\gamma$ when $\epsilon = 1$ and $\tau = 0.02$. When $\epsilon = 1$ in the original RR mechanism, we computed the estimation error before the attack as 0.78. Increasing $\gamma$ results in eliminating less states by the attacker. For instance, if attacker selects $\gamma = 0.5$, it cannot eliminate any states and the estimation is still $0.78$. As $\gamma$ decreases, more states are eliminated and the estimation error keeps decreasing up to a point (up to $\gamma=0.03$ in our experiments, which provides the smallest estimation error). 
As we further decreased $\gamma$ beyond this point, we observed higher estimation error values, since as $\gamma$ approaches to $0$, all $3$ states are eliminated for more SNPs. Also, when $\gamma = 0$, we computed the estimation error as 0.78 as well. We also observed similar results for different values of $\epsilon$. Similarly, when $\gamma = 0.03$, we obtained the smallest estimation error for the attacker (and hence the strongest attack) when $\tau=0.02$.

\begin{table}
\small
\centering
\caption{Privacy (in terms of estimation error) and utility (in terms of accuracy of beacon responses) of the proposed scheme for different values of $\hat{\tau}$ (correlation threshold) when $\hat{\gamma} = 0.03$ and $\epsilon = 1$. Estimation error is computed by assuming the attacker uses $\tau = 0.02$ and $\gamma = 0.03$.}
\begin{tabular}{ c|c|c|c|c|c }
\hline
$\hat{\tau}$ & 0.02 & 0.04 & 0.06 & 0.08 & 0.1 \\
\hline
Estimation Error ($E$) & 0.483 & 0.486 & 0.492 & 0.499 & 0.503 \\
\hline
Accuracy ($A$) & 0.950 & 0.942 & 0.918 & 0.892 & 0.865 \\
\hline
\end{tabular}
\label{table:tauEffect}
\end{table}

\begin{table}
\small
\centering
\caption{Privacy (in terms of estimation error) and utility (in terms of accuracy of beacon responses) of the proposed scheme for different values of $\hat{\gamma}$ (inconsistency threshold) when $\hat{\tau} = 0.02$ and $\epsilon = 1$. Estimation error is computed by assuming the attacker uses $\tau = 0.02$ and $\gamma = 0.03$.}
\begin{tabular}{ c|c|c|c|c|c }
\hline
$\hat{\gamma}$ & 0.01 & 0.02 & 0.03 & 0.04 & 0.05 \\
\hline
Estimation Error ($E$) & 0.490 & 0.487 & 0.483 & 0.479 & 0.476\\
\hline
Accuracy ($A$) & 0.932 & 0.940 & 0.950 & 0.954 & 0.959 \\
\hline
\end{tabular}
\label{table:gammaEffect}
\end{table}

Since the attack against the original RR mechanism is the strongest when $\tau=0.02$ and $\gamma=0.03$, we set the correlation parameters of the proposed data sharing algorithm the same as the attack parameters (i.e., $\hat{\tau} = 0.02$ and $\hat{\gamma} = 0.03$) in Section~\ref{subsec:compareLDP}. Here, we study the effect of changing these parameters on the performance of the proposed mechanism. We assume that the attacker does not know the parameters ($\hat{\tau}$ and $\hat{\gamma}$) used in the algorithm and uses the parameters providing the strongest attack ($\tau = 0.02$ and $\gamma = 0.03$) against the original RR mechanism. First, we evaluated the effect of correlation threshold $\hat{\tau}$ on privacy and utility (all correlations that are smaller than $\hat{\tau}$ are considered as low by the algorithm). 
Our results are shown in Table~\ref{table:tauEffect}. We observed that increasing $\hat{\tau}$ increases the attacker's estimation error since we assume the attacker does not know $\hat{\tau}$ and uses $\tau = 0.02$ in its attack. However, using $\hat{\tau} = 0.02$ provided the best utility for the proposed algorithm. 
Since there is no correlation (conditional probability) that is less than $0.02$ in our dataset, the minimum possible value that we can use for $\hat{\tau}$ in the algorithm is $0.02$. We also show the privacy and utility of the proposed scheme for different values of $\hat{\gamma}$ in Table~\ref{table:gammaEffect}. We observed that increasing $\hat{\gamma}$ slightly increases utility, however, the privacy also decreases at the same time.

In the previous experiments (Table~\ref{table:tauEffect} and Table~\ref{table:gammaEffect}), we assumed that the attacker does not know the parameters used in the experiments and uses $\tau = 0.02$ and $\gamma = 0.03$. However, the attacker can perform stronger attacks if it knows the design parameters ($\hat{\tau}$ and $\hat{\gamma}$) of the algorithm. Thus, we also computed the attacker's estimation error by assuming it knows the parameters used in the algorithm ($\hat{\tau} = 0.02$ and $\hat{\gamma} = 0.03)$). Estimation error of the attacker for different values of $\tau$ and $\gamma$ are shown in Tables~\ref{table:tauEffectAttacker} and~\ref{table:gammaEffectAttacker}, respectively. When we increased $\tau$ up to 0.1, we observed a slight decrease in the estimation error. For instance, when $\tau = 0.1$ and $\hat{\tau} = 0.02$, we observed the estimation error of the attacker as $0.42$. Similarly, the attacker can decrease the estimation error to $0.434$ by knowing the value of $\hat{\gamma}$ and selecting $\gamma = 0.01$. We also observed that for $\tau$ values greater than $0.1$ and $\gamma$ values less than 0.01, the decrease in attacker's estimation error converged. Overall, we conclude that the attacker can slightly reduce its estimation error by knowing the design parameters of the proposed mechanism, however, the gain of the attacker (in terms of reduced estimation error) is negligible (at most 0.07). Furthermore, the proposed scheme still preserves its advantage over the original RR mechanism in all considered scenarios. These results show that varying design parameters only slightly affect the performance of the proposed scheme.

In our experiments, we assume that the attacker has the same background knowledge (i.e., correlations between SNPs) as the data owner. If the attacker's knowledge is weaker than this assumption (e.g., if the computed correlations on the attacker's side are not accurate), then its attack will be less successful and its estimation error will be higher than the one we computed in our experiments. On the other hand, if the attacker's knowledge about the correlations in the data is stronger than the data owner, it can perform more successful attacks. To validate this, we added noise to the correlations computed by the data owner and observed that the attacker obtains a lower estimation error than the one in our experiments. In the worst case scenario, when the data owner does not know (or use) the correlations in the data, the estimation error of the attacker becomes equal to its estimation error when it performs the attack to the original RR mechanism (i.e., solid line marked with triangles in Figure \ref{fig:ldpComp}).

\begin{table}
\small
\centering
\caption{Privacy (in terms of estimation error) of the proposed scheme for different values of $\tau$ (correlation threshold of the attacker) when $\gamma = 0.03$ and $\epsilon = 1$. The parameters used in the data sharing algorithm are $\hat{\tau} = 0.02$ and $\hat{\gamma} = 0.03$.}
\begin{tabular}{ c|c|c|c|c|c }
\hline
$\tau$ & 0.02 & 0.04 & 0.06 & 0.08 & 0.10 \\
\hline
Estimation Error ($E$) & 0.483 & 0.478 & 0.462 & 0.446 & 0.420\\
\hline
\end{tabular}
\label{table:tauEffectAttacker}
\end{table}

\begin{table}
\small
\centering
\caption{Privacy (in terms of estimation error) of the proposed scheme for different values of $\gamma$ (inconsistency threshold of the attacker) when $\tau = 0.02$ and $\epsilon = 1$. The parameters used in the data sharing algorithm are $\hat{\tau} = 0.02$ and $\hat{\gamma} = 0.03$.}
\begin{tabular}{ c|c|c|c|c|c }
\hline
$\gamma$ & 0.01 & 0.02 & 0.03 & 0.04 & 0.05 \\
\hline
Estimation Error ($E$) & 0.434 & 0.468 & 0.483 & 0.497 & 0.508\\
\hline
\end{tabular}
\label{table:gammaEffectAttacker}
\end{table}

\subsection{The Effect of the Processing Order on Utility}
\label{subsec:orderEval}
In this section, we show the effect of different order of processing on the utility of the beacon responses. For all experiments, we set the parameters the same as in Section~\ref{subsec:compareLDP} (i.e., $\hat{\tau}=0.02$ and $\hat{\gamma}=0.03$) and we also quantified the accuracy in terms of the fraction of correct beacon responses for a population. We reported the results averaged over 100 trials.

To demonstrate that the greedy order of processing (in Section~\ref{sec:heuristic}) outperforms the random order 
and provides an accuracy that is close to the optimal order (in Algorithm~\ref{Value_Iteration}), we first compared them using a small dataset of 10 SNPs of 10 individuals (obtained from the same HapMap dataset~\cite{international2003international} introduced before).
When processing the SNPs of an individual $I_j$ using the random order, we randomly permuted the order of her SNP sequence and then fed it into Algorithm~\ref{alg:alg1}. 
Assuming each donor has the same privacy budget ($\epsilon$) and varying the privacy budget from $0.2$ to $2$, we show the results in
Figure~\ref{fig:utility_vs_e_small}. 
We observed that for all the privacy budgets, the accuracy obtained by the greedy order is close to that obtained by the optimal order (when $\epsilon\geq 1$, the accuracy provided by both orders differ only by less than $2\%$). Whereas, the accuracy achieved by the random order is the lowest for all the privacy budgets because the random order does not try to maximize individuals' expected beacon utility. These results show that greedy order of processing (in Section~\ref{sec:heuristic}) performs comparable to the optimal algorithm (in Algorithm~\ref{Value_Iteration}), and hence we use the greedy algorithm for our evaluations with larger datasets. 
\begin{figure}
    \centering
    \includegraphics[width=8cm,keepaspectratio]{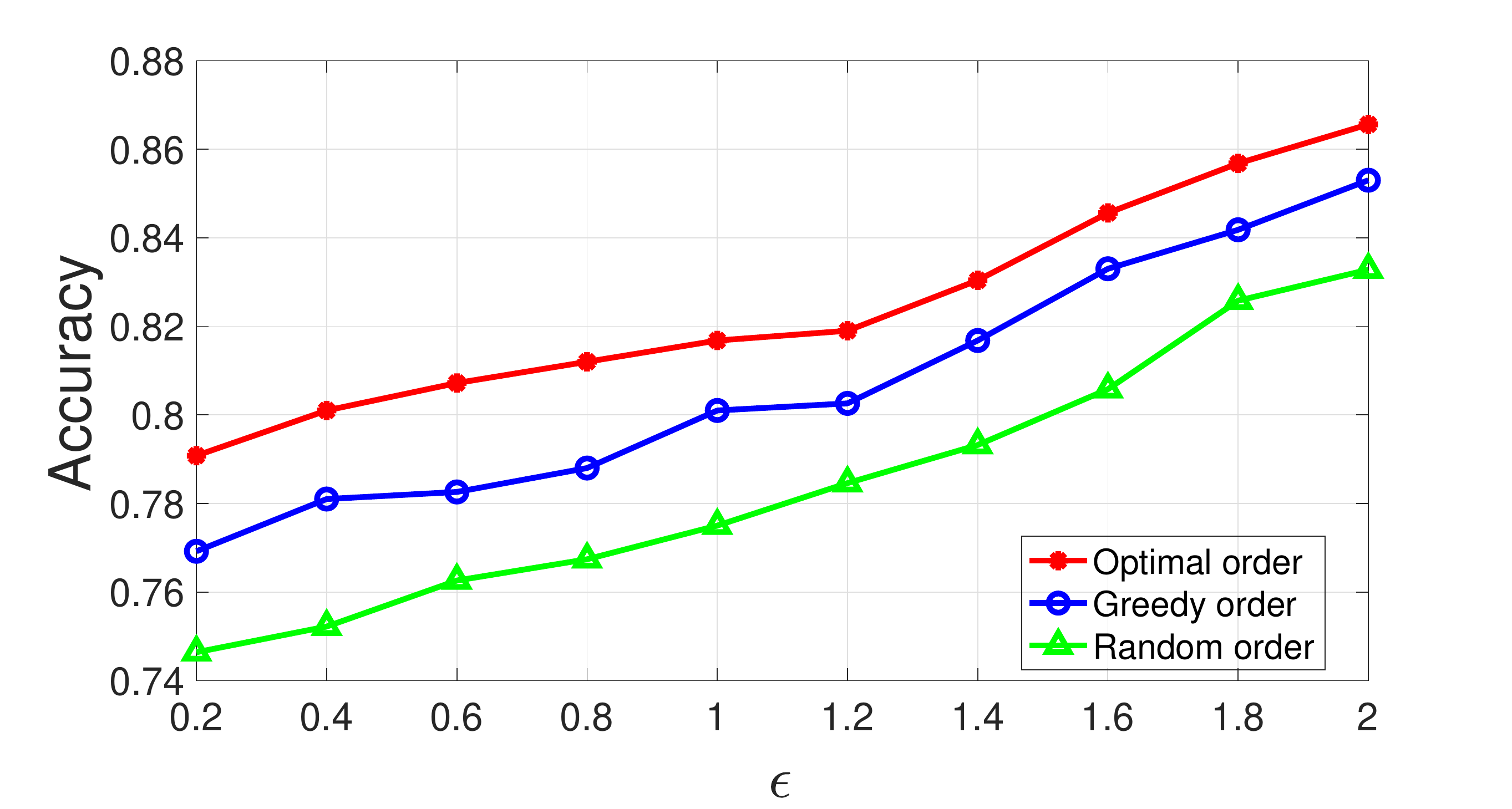}
    \caption{Accuracy of beacon responses on 10 SNPs from 10 individuals using optimal (Algorithm~\ref{Value_Iteration}), greedy (Section~\ref{sec:heuristic}), and random orders of processing.} 
    \label{fig:utility_vs_e_small}
\end{figure}

Next, we compared the accuracy achieved by the greedy and random orders on the original dataset (i.e., 1000 SNPs of 156 individuals). 
The experiment results are shown in Figure~\ref{fig:utility_vs_e_large}. We observed that compared to the small dataset, the accuracy is improved significantly. For example, even under very limited privacy budgets (e.g., $\epsilon\leq 0.4$), both orders can achieve an accuracy over $93\%$ since large dataset contains stronger (and more) correlations among SNPs. Correlations in the data is critical for the utility of the proposed data sharing mechanism, since when data is correlated, the proposed algorithm eliminates statistically unlikely states and adjusts the probability distributions of the remaining states in such a way that deviating highly from the ``useful values" of the shared SNPs is small (as discussed in Section~\ref{subsec:improveUtility}). 
From Figure~\ref{fig:utility_vs_e_large}, we also observed that the accuracy achieved by the greedy order consistently outperforms that obtained by the random order. This suggests that the utility varies under different processing orders and we can improve the utility of shared data points (SNPs) in a strategic way (e.g., by selecting them in a greedy manner). This outcome can also be generalized when sharing other types of correlated data. Another advantage of determining the processing order using the greedy algorithm is its computational complexity ($\mathcal{O}(l^2)$, where $l$ is the number of shared SNPs of a donor), whereas the computational complexity of the optimal algorithm (in Algorithm~\ref{Value_Iteration}) is $\mathcal{O}(3^l)$. 
\begin{figure}
    \centering
    \includegraphics[width=8cm,keepaspectratio]{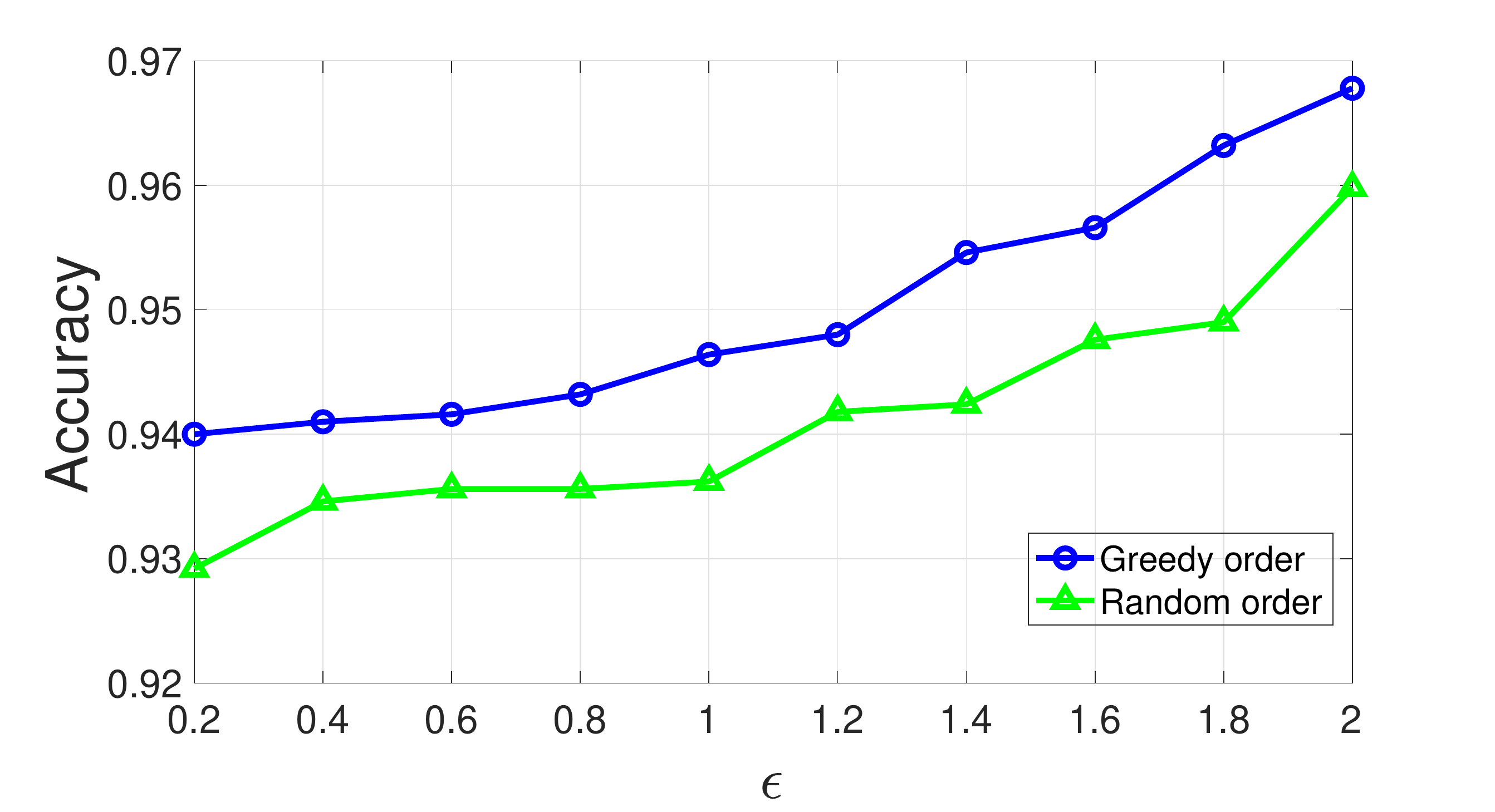}
    \caption{Accuracy of beacon responses on the original dataset using greedy and random orders of processing.} 
    \label{fig:utility_vs_e_large}
\end{figure}

\section{Discussion}
\label{sec:disc}
In this section, we discuss how to consider kinship in data sharing and how attacker's background knowledge affects the privacy guarantees.
\subsection{Selection of Privacy Budget by Considering Kinship}
\label{sec:discussion}

Due to rules of inheritance (i.e., Mendel's law), sharing the value of a SNP also (indirectly) reveals information about genomes of donor's family members and this may help an attacker to improve its inference about some SNPs of the family members. Let a genome donor $I_j$ share her $SNP_i$ ($x_i^j$) as $y_i^j$ and $I_f$ be one of her family members. Here, we discuss the information gain of an attacker about $SNP_i$ of the family member $I_f$ ($x_i^f$) by receiving $y_i^j$ in terms of the privacy budget of $I_j$. We assume that all family members have their own privacy budgets (e.g., $\epsilon_j$ being the privacy budget for $I_j$) that they do not want to violate. An attacker can gain information about genomes of individuals by using the shared SNPs of their family members. Hence, we first compute the indirect privacy budget of a family member $I_f$ when $I_j$ shares her $SNP_i$. Then, we propose an algorithm to compute the maximum privacy budget of a genome donor to preserve the privacy of her family members by considering their privacy budgets and the previously shared SNPs.

\subsubsection*{Computing the Attacker's Information Gain}

Since the proposed data sharing mechanism assigns different probability of sharing for a given SNP under different scenarios (as discussed in Section~\ref{subsec:improveUtility}), we discuss the privacy of family members by assuming all three states of the SNP are possible (i.e., no state is eliminated due to correlations) for the sake of generality. The computations in this part can be done similarly for other scenarios (e.g., the ones having less possible states) as well using the probability distributions in Figure~\ref{fig:prob2}. In the scenario with three possible states, when the donor $I_j$ shares $0$ as the value of $SNP_i$, the original value of the SNP is $0$ with probability $p$ ($\Pr(x_i^j = 0~|~y_i^j = 0) = p$). Also, both $1$ or $2$ can also be the original value of that SNP with probability $q$ ($\Pr(x_i^j = 0~|~y_i^j = 1) = \Pr(x_i^j = 0~|~y_i^j = 2) = q$). As mentioned before, to achieve $(\epsilon,T)$-dependent LDP, $p$ and $q$ are selected as $e^\epsilon_j / (e^\epsilon_j + 2)$ and $1 / (e^\epsilon_j + 2)$, respectively. Therefore, using the Mendel's law, the attacker can compute $\Pr(x_i^f = d_{\alpha}~|~y_i^j = d_{\beta})$ for $d_{\alpha}, d_{\beta} \in \{0,1,2\}$. 
For instance, $\Pr(x_i^f = 0~|~y_i^j = 0)$ can be computed by the attacker as $\sum_{m \in \{0,1,2\}} \Pr(x_i^f = 0~|~x_i^j = m) \Pr(x_i^j = m~|~y_i^j = 0)$. This sum is also equal to: $\Pr(x_i^f = 0~|~x_i^j = 0)p + \Pr(x_i^f = 0~|~x_i^j = 1)q + \Pr(x_i^f = 0~|~x_i^j = 2)q$. Similarly, the attacker can also compute $\Pr(x_i^f = d_{\alpha}~|~y_i^j = d_{\beta})$ for other $d_{\alpha}$ and $d_{\beta}$ values. Hence, based on the conditional probabilities $\Pr(x_i^f = d_{\alpha}~|~x_i^j = d_{\beta})$ (computed using Mendel's law), the attacker can gain a certain amount of information about each family member of the donor (the amount of information depends of the kinship relationship between the donor and the corresponding family member). Moreover, the attacker can gain more information about a victim if more than one of victim's family members share their values for the same SNP. 

As discussed in Section~\ref{subsec:genomics}, each SNP of a child inherits one allele from the mother and one allele from the father. This means, the attacker can gain the most information from the first degree family members. For instance, if a child has $0$ as the value of $SNP_i$, the attacker can (using Mendel's law) infer that neither of her parents can have $2$ as the value of $SNP_i$. In Appendix~\ref{app:kinship1}, we analyze the privacy loss of a victim when one of her/his first degree relatives share her $SNP_i$ using the proposed method. We also extend the analysis and consider a case, in which two children of a victim (parent) share their SNPs under $(\epsilon,T)$-dependent LDP in Appendix~\ref{app:kinship2}. 

\subsubsection*{Determining the Maximum Privacy Budget}
\label{sec:maxprivacy}

$\epsilon$ value of each family member who wants to share her SNPs can be computed similarly by considering the privacy budgets of the family members and the family members who previously shared their SNPs under $(\epsilon,T)$-dependent LDP. Let all members in a family denoted as $F_1$, \ldots, $F_m$ and assume $n$ of them ($F_1$, \ldots,$F_n$) shared their SNPs previously. A basic algorithm is given in Algorithm~\ref{epsilonAlgorithm} to compute the maximum privacy budget $\epsilon_{max_i}$ that can be used by a family member $F_{n+1}$ who wants to share her $SNP_i$. For each family member $F_s$ who did not share $SNP_i$, the algorithm computes maximum privacy budget that can be used by $F_{n+1}$ to preserve privacy of $F_s$ (by computing $\Pr(x_{i}^{F_s} = b~|~y_i^{F_1} \wedge ... \wedge y_i^{F_n} \wedge (y_i^{F_{n+1}} = a))$ using the privacy budgets of the family members who previously shared their data ($\epsilon_{F_1}$, \ldots, $\epsilon_{F_n}$) and solving the equation in Line 9). Then, minimum of these values is returned as the maximum privacy budget $\epsilon_{max_i}$ for sharing $SNP_i$. Assuming $F_{n+1}$ wants to share all her $l$ SNPs, she can run Algorithm~\ref{epsilonAlgorithm} for all of her SNPs ($SNP_1$,...,$SNP_l$) and obtain $\epsilon_{max_1}$,...,$\epsilon_{max_l}$. At the end, $F_{n+1}$ selects the minimum of these $\epsilon$ values as her privacy budget ($\epsilon_{max} \leftarrow min(\epsilon_{max_1},...,\epsilon_{max_l})$). Then, if $\epsilon_{max} \geq \epsilon_{F_{n+1}}$, $F_{n+1}$ uses $\epsilon_{F_{n+1}}$ as her maximum privacy budget. However, if $\epsilon_{max} < \epsilon_{F_{n+1}}$, $F_{n+1}$ should use $\epsilon_{max}$ in the proposed data sharing algorithm to preserve the privacy of the other family members while sharing her data.

\begin{algorithm}
\footnotesize
\SetKwInOut{Input}{input}
\SetKwInOut{Output}{output}
\Input{Privacy budgets of all family members, previously shared data by $n$ family members $y_i^{F_1}$, $y_i^{F_2}$,...., $y_i^{F_n}$, privacy budgets in previous sharings $\epsilon_{F_1}$, $\epsilon_{F_2}$, ..., $\epsilon_{F_n}$}
\Output{Maximum privacy budget $\epsilon_{max_i}$ to preserve the privacy of family members}
$\epsilon_{max_i} \leftarrow \infty$;
\ForAll{family member $F_s$ who did not share her $SNP_i$ }
{
    $e \leftarrow [0,0,0]$; 
    
    \ForAll{$a \in \{0,1,2\}$}
    {
        $z \leftarrow [0,0,0]$;
        
        \ForAll{$b \in \{0,1,2\}$}
        {
        $z[b] \leftarrow \Pr(x_{i}^{F_s} = b~|~y_i^{F_1} \wedge ... \wedge y_i^{F_n} \wedge (y_i^{F_{n+1}} = a))$;
        }
        $e[a] \leftarrow \textrm{solve}~ (\frac{max(z)}{ min(z)} = e^{\epsilon_{F_s}}) $;
    }
    
    \If{$max(e) < \epsilon_{max_i}$} { $\epsilon_{max_i} \longleftarrow max(e)$;
	     	}
}
	\caption{Determining the maximum privacy budget $\epsilon_{max_i}$ of $F_{n+1}$ for sharing $SNP_i$ to preserve the privacy of family members.}

\label{epsilonAlgorithm}
\end{algorithm}

Note that in Algorithm~\ref{epsilonAlgorithm}, we only consider the privacy budgets of the family members who did not share their SNPs. Since $F_1$, \ldots, $F_n$ shared their SNPs before $F_{n+1}$, we assume that they already used all or most of their privacy budgets. Hence, their privacy budgets will probably be exceeded with another sharing. If we also consider the privacy budgets of these $n$ family members (who shared their genomes before $F_{n+1}$), the algorithm will most likely return $0$ as $\epsilon_{max_{F_{n+1}}}$. Thus, we can assume that by sharing data, an individual accepts an indirect increase in her privacy budget when other family members share their data in the future. Furthermore, practical selection of the privacy parameter ($\epsilon$) by a donor and coordination among family members during data sharing require further investigation and they are beyond the scope of this paper.

\subsection{Privacy Guarantees Against Attacks At Large }
\label{sec:leakage}
We have demonstrated that the proposed mechanism can achieve the $(\epsilon,T)$-dependent LDP for each SNP when an attacker conducts a correlation attack by checking the publicly available pairwise conditional probabilities of the SNPs. However, in practice, different attackers with various prior knowledge may conduct more sophisticated attacks using advanced machine learning techniques (e.g., using Bayesian inference~\cite{yang2015bayesian,liu2016dependence}). Here, we will discuss the privacy guarantees of the proposed mechanism on attacks at large.     

Considering the case, in which an individual $I_j$ shares her $l$ SNPs in $\mathcal{X}^j$, we characterize an arbitrary attacker ($\mathbb{A}$) by its prior knowledge tuple ($\mathbb{A}_k$), i.e., $\mathbb{A}_k = \{\mathcal{R},PC_{\mathcal{Q}}\}$, 
where $\mathcal{R}$ is the set of SNPs whose true states are already known by the attacker ($\mathcal{R}\subset \mathcal{X}^j, |\mathcal{R}|< l$), and  $PC_{\mathcal{Q}}$ is the pairwise conditional probabilities of the set of SNPs in $\mathcal{X}^j$ that are known by the attacker ($\mathcal{Q}\subset \mathcal{X}^j, |\mathcal{Q}|\leq l$). In Section~\ref{subsec:attack}, the considered attacker is characterized as
$\mathbb{A}_k=\{\emptyset, PC_{\mathcal{X}^j}\}$. 
Here, we investigate the impact of more general attackers on the privacy leakage of a single SNP.

We define the privacy leakage of a SNP $x_i\in \mathcal{X}^j$ caused by the attacker with prior knowledge $\mathbb{A}_k$ as $\text{PrvcLkg}_{x_i}^{\mathbb{A}_k} = \sup_{\alpha\in\{0,1,2\}}\Pr(x_i = \alpha|\mathcal{Y}^j,\mathbb{A}_k)$, where $\mathcal{Y}^j$ is the released SNPs via Algorithm~\ref{alg:alg1}. $\text{PrvcLkg}_{x_i}^{\mathbb{A}_k}$ covers a wide-range of attacks using machine learning techniques, because most of the learning algorithms give outputs in the form of  posterior probabilities, such as Bayesian inference and deep learning. Then, we have
$\text{PrvcLkg}_{x_i}^{\mathbb{A}_k} = \sup_{\alpha\in\{0,1,2\}}\Pr(x_i = \alpha|\mathcal{Y}^j, \mathbb{A}_k)\\
 = \sup_{\alpha\in\{0,1,2\}} \frac{  \Pr(\mathcal{Y}^j| x_i = \alpha  ,\mathbb{A}_k)  \Pr(x_i = \alpha  ,\mathbb{A}_k) }{\Pr(\mathcal{Y}^j, \mathbb{A}_k)}\\
 = \sup_{\alpha\in\{0,1,2\}} \frac{\Pr(\mathcal{Y}^j| x_i = \alpha  ,\mathbb{A}_k) }{\Pr(\mathcal{Y}^j|\mathbb{A}_k)}   \Pr(x_i = \alpha  |\mathbb{A}_k) \\
 = \sup_{\alpha\in\{0,1,2\}} \boxed{\frac{\Pr(\mathcal{Y}^j| x_i = \alpha  ,\mathbb{A}_k) }{\Pr(\mathcal{Y}^j|     x_i = \beta,    \mathbb{A}_k)}   }   
\frac{\Pr(\mathcal{Y}^j|     x_i = \beta,    \mathbb{A}_k)}{\Pr(\mathcal{Y}^j|     \mathbb{A}_k)}  \Pr(x_i = a  |\mathbb{A}_k)\\
\stackrel{(*)}= \sup_{\alpha\in\{0,1,2\}}  \boxed{e^\epsilon}     \frac{  \Pr(\mathcal{Y}^j ,    x_i = \beta,    \mathbb{A}_k)    \Pr(\mathbb{A}_k)     }{ \Pr(x_i = \beta,    \mathbb{A}_k) \Pr(\mathcal{Y}^j ,    \mathbb{A}_k)      }  \Pr(x_i = a  |\mathbb{A}_k)\\
 =  \sup_{\alpha\in\{0,1,2\}}  e^\epsilon \boxed{\Pr(x_i = \beta|\mathcal{Y}^j, \mathbb{A}_k)} \frac{\Pr(x_i = a  |\mathbb{A}_k)}{\Pr(x_i = \beta|    \mathbb{A}_k)}      \\
 \stackrel{(\#)}  \leq  e^\epsilon \boxed{(1-\text{PrvcLkg}_{x_i}^{\mathbb{A}_k})}\zeta$,\\
 where $(*)$ follows from the definition of $(\epsilon,T)$-dependent LDP, and $\epsilon$ is determined after eliminating the statistically unlike states of $x_i$, and in $(\#)$, $\zeta = \frac{\Pr(x_i = \alpha  |\mathbb{A}_k)}{\Pr(x_i = \beta|\mathbb{A}_k)}$ is the ratio between prior probabilities, which is a prior knowledge-specific parameter.    $\text{PrvcLkg}_{x_i}^{\mathbb{A}_k} \leq e^\epsilon (1-\text{PrvcLkg}_{x_i}^{\mathbb{A}_k})\zeta$ can further be simplified as $\text{PrvcLkg}_{x_i}^{\mathbb{A}_k} \leq   \max\{\frac{1}{\zeta e^\epsilon +1},   \frac{\zeta e^\epsilon}{\zeta e^\epsilon +1}   \}$. As a result, we provide an upper bound for the privacy leakage of an arbitrary SNP caused by an arbitrary attacker, who has a specific prior knowledge. This suggests that an attacker can only have impact on the privacy leakage via the ratio between prior probabilities, i.e., $\zeta$, which, however, is independent of his adopted inference technique (attacking scheme). Consequently, under general attacks, the proposed mechanism guarantees that the probability of correctly inferring the state of an arbitrary unknown SNP  will not exceed $\max\{\frac{1}{\zeta e^\epsilon +1},   \frac{\zeta e^\epsilon}{\zeta e^\epsilon +1}   \}$. 
\section{Conclusion}
\label{sec:conclusion}
In this paper, we have introduced $(\epsilon,T)$-dependent LDP and proposed a data sharing scheme for genomic data sharing achieving $(\epsilon,T)$-dependent LDP. We have first described a correlation attack to show \color{black} that directly applying the randomized response mechanism to correlated data causes vulnerabilities. To improve privacy against the correlation attacks, we have proposed a scheme that eliminates certain states of a SNP (and does not use such states during data sharing) which are loosely correlated with the previously shared SNPs. The proposed scheme decides a value to share among the non-eliminated states by providing formal privacy guarantees. To improve the utility of the shared data, we have shown how to adjust probability distributions for the non-eliminated states of the SNPs while still guaranteeing $(\epsilon,T)$-dependent LDP. We have also proposed an optimal algorithm and a greedy algorithm to determine the processing order of SNPs in the proposed data sharing algorithm to optimize utility. Furthermore, we have discussed the effect of genomic data sharing on family members and proposed an algorithm to decide the privacy budget of a genome donor by considering the privacy preferences of her family members. We have implemented the proposed scheme and evaluated its privacy and utility via experiments on a real-life genomic dataset. The proposed data sharing mechanism can also be utilized for sharing of similar sensitive information that includes correlations (e.g., location patterns). In future work, we will evaluate the proposed mechanism considering different application of the data collector. We will also study how to compute the data sharing probabilities for different values of the SNP as a donor shares data with more data collectors.

\bibliographystyle{ACM-Reference-Format}
\bibliography{sample-base}


\begin{thebibliography}{41}


\ifx \showCODEN    \undefined \def \showCODEN     #1{\unskip}     \fi
\ifx \showDOI      \undefined \def \showDOI       #1{#1}\fi
\ifx \showISBNx    \undefined \def \showISBNx     #1{\unskip}     \fi
\ifx \showISBNxiii \undefined \def \showISBNxiii  #1{\unskip}     \fi
\ifx \showISSN     \undefined \def \showISSN      #1{\unskip}     \fi
\ifx \showLCCN     \undefined \def \showLCCN      #1{\unskip}     \fi
\ifx \shownote     \undefined \def \shownote      #1{#1}          \fi
\ifx \showarticletitle \undefined \def \showarticletitle #1{#1}   \fi
\ifx \showURL      \undefined \def \showURL       {\relax}        \fi
\providecommand\bibfield[2]{#2}
\providecommand\bibinfo[2]{#2}
\providecommand\natexlab[1]{#1}
\providecommand\showeprint[2][]{arXiv:#2}

\bibitem[\protect\citeauthoryear{Ayday, Raisaro, Hubaux, and Rougemont}{Ayday
  et~al\mbox{.}}{2013}]%
        {related:ermanclinic}
\bibfield{author}{\bibinfo{person}{Erman Ayday}, \bibinfo{person}{Jean~Louis
  Raisaro}, \bibinfo{person}{Jean-Pierre Hubaux}, {and}
  \bibinfo{person}{Jacques Rougemont}.} \bibinfo{year}{2013}\natexlab{}.
\newblock \showarticletitle{Protecting and evaluating genomic privacy in
  medical tests and personalized medicine}. In
  \bibinfo{booktitle}{\emph{Proceedings of the 12th ACM workshop on Workshop on
  privacy in the electronic society}}. ACM, \bibinfo{pages}{95--106}.
\newblock


\bibitem[\protect\citeauthoryear{Baldi, Baronio, De~Cristofaro, Gasti, and
  Tsudik}{Baldi et~al\mbox{.}}{2011}]%
        {related:baldi}
\bibfield{author}{\bibinfo{person}{Pierre Baldi}, \bibinfo{person}{Roberta
  Baronio}, \bibinfo{person}{Emiliano De~Cristofaro}, \bibinfo{person}{Paolo
  Gasti}, {and} \bibinfo{person}{Gene Tsudik}.}
  \bibinfo{year}{2011}\natexlab{}.
\newblock \showarticletitle{Countering gattaca: efficient and secure testing of
  fully-sequenced human genomes}. In \bibinfo{booktitle}{\emph{Proceedings of
  the 18th ACM conference on Computer and communications security}}. ACM,
  \bibinfo{pages}{691--702}.
\newblock


\bibitem[\protect\citeauthoryear{Bassily, Nissim, Stemmer, and
  Thakurta}{Bassily et~al\mbox{.}}{2017}]%
        {bassily2017practical}
\bibfield{author}{\bibinfo{person}{Raef Bassily}, \bibinfo{person}{Kobbi
  Nissim}, \bibinfo{person}{Uri Stemmer}, {and} \bibinfo{person}{Abhradeep~Guha
  Thakurta}.} \bibinfo{year}{2017}\natexlab{}.
\newblock \showarticletitle{Practical locally private heavy hitters}. In
  \bibinfo{booktitle}{\emph{Advances in Neural Information Processing
  Systems}}. \bibinfo{pages}{2288--2296}.
\newblock


\bibitem[\protect\citeauthoryear{Bertsekas, Bertsekas, Bertsekas, and
  Bertsekas}{Bertsekas et~al\mbox{.}}{1995}]%
        {bertsekas1995dynamic}
\bibfield{author}{\bibinfo{person}{Dimitri~P Bertsekas},
  \bibinfo{person}{Dimitri~P Bertsekas}, \bibinfo{person}{Dimitri~P Bertsekas},
  {and} \bibinfo{person}{Dimitri~P Bertsekas}.}
  \bibinfo{year}{1995}\natexlab{}.
\newblock \bibinfo{booktitle}{\emph{Dynamic programming and optimal control}}.
  Vol.~\bibinfo{volume}{1}.
\newblock \bibinfo{publisher}{Athena scientific Belmont, MA}.
\newblock


\bibitem[\protect\citeauthoryear{Boutilier, Dearden, and Goldszmidt}{Boutilier
  et~al\mbox{.}}{2000}]%
        {boutilier2000stochastic}
\bibfield{author}{\bibinfo{person}{Craig Boutilier}, \bibinfo{person}{Richard
  Dearden}, {and} \bibinfo{person}{Mois{\'e}s Goldszmidt}.}
  \bibinfo{year}{2000}\natexlab{}.
\newblock \showarticletitle{Stochastic dynamic programming with factored
  representations}.
\newblock \bibinfo{journal}{\emph{Artificial intelligence}}
  \bibinfo{volume}{121}, \bibinfo{number}{1-2} (\bibinfo{year}{2000}),
  \bibinfo{pages}{49--107}.
\newblock


\bibitem[\protect\citeauthoryear{Cao, Yoshikawa, Xiao, and Xiong}{Cao
  et~al\mbox{.}}{2017}]%
        {cao2017quantifying}
\bibfield{author}{\bibinfo{person}{Yang Cao}, \bibinfo{person}{Masatoshi
  Yoshikawa}, \bibinfo{person}{Yonghui Xiao}, {and} \bibinfo{person}{Li
  Xiong}.} \bibinfo{year}{2017}\natexlab{}.
\newblock \showarticletitle{Quantifying differential privacy under temporal
  correlations}. In \bibinfo{booktitle}{\emph{2017 IEEE 33rd International
  Conference on Data Engineering (ICDE)}}. IEEE, \bibinfo{pages}{821--832}.
\newblock


\bibitem[\protect\citeauthoryear{Chanyaswad, Dytso, Poor, and
  Mittal}{Chanyaswad et~al\mbox{.}}{2018}]%
        {chanyaswad2018mvg}
\bibfield{author}{\bibinfo{person}{Thee Chanyaswad}, \bibinfo{person}{Alex
  Dytso}, \bibinfo{person}{H~Vincent Poor}, {and} \bibinfo{person}{Prateek
  Mittal}.} \bibinfo{year}{2018}\natexlab{}.
\newblock \showarticletitle{Mvg mechanism: Differential privacy under
  matrix-valued query}. In \bibinfo{booktitle}{\emph{Proceedings of the 2018
  ACM SIGSAC Conference on Computer and Communications Security}}.
  \bibinfo{pages}{230--246}.
\newblock


\bibitem[\protect\citeauthoryear{Cheu, Smith, Ullman, Zeber, and Zhilyaev}{Cheu
  et~al\mbox{.}}{2019}]%
        {cheu2019distributed}
\bibfield{author}{\bibinfo{person}{Albert Cheu}, \bibinfo{person}{Adam Smith},
  \bibinfo{person}{Jonathan Ullman}, \bibinfo{person}{David Zeber}, {and}
  \bibinfo{person}{Maxim Zhilyaev}.} \bibinfo{year}{2019}\natexlab{}.
\newblock \showarticletitle{Distributed differential privacy via shuffling}. In
  \bibinfo{booktitle}{\emph{Annual International Conference on the Theory and
  Applications of Cryptographic Techniques}}. Springer,
  \bibinfo{pages}{375--403}.
\newblock


\bibitem[\protect\citeauthoryear{Consortium et~al\mbox{.}}{Consortium
  et~al\mbox{.}}{2003}]%
        {international2003international}
\bibfield{author}{\bibinfo{person}{International~HapMap Consortium}
  {et~al\mbox{.}}} \bibinfo{year}{2003}\natexlab{}.
\newblock \showarticletitle{The international HapMap project}.
\newblock \bibinfo{journal}{\emph{Nature}} \bibinfo{volume}{426},
  \bibinfo{number}{6968} (\bibinfo{year}{2003}), \bibinfo{pages}{789}.
\newblock


\bibitem[\protect\citeauthoryear{Cormode, Kulkarni, and Srivastava}{Cormode
  et~al\mbox{.}}{2018}]%
        {cormode2018marginal}
\bibfield{author}{\bibinfo{person}{Graham Cormode}, \bibinfo{person}{Tejas
  Kulkarni}, {and} \bibinfo{person}{Divesh Srivastava}.}
  \bibinfo{year}{2018}\natexlab{}.
\newblock \showarticletitle{Marginal release under local differential privacy}.
  In \bibinfo{booktitle}{\emph{Proceedings of the 2018 International Conference
  on Management of Data}}. \bibinfo{pages}{131--146}.
\newblock


\bibitem[\protect\citeauthoryear{Cormode, Kulkarni, and Srivastava}{Cormode
  et~al\mbox{.}}{2019}]%
        {cormode2019answering}
\bibfield{author}{\bibinfo{person}{Graham Cormode}, \bibinfo{person}{Tejas
  Kulkarni}, {and} \bibinfo{person}{Divesh Srivastava}.}
  \bibinfo{year}{2019}\natexlab{}.
\newblock \showarticletitle{Answering range queries under local differential
  privacy}.
\newblock \bibinfo{journal}{\emph{Proceedings of the VLDB Endowment}}
  \bibinfo{volume}{12}, \bibinfo{number}{10} (\bibinfo{year}{2019}),
  \bibinfo{pages}{1126--1138}.
\newblock


\bibitem[\protect\citeauthoryear{Dean and Givan}{Dean and Givan}{1997}]%
        {dean1997model}
\bibfield{author}{\bibinfo{person}{Thomas Dean} {and} \bibinfo{person}{Robert
  Givan}.} \bibinfo{year}{1997}\natexlab{}.
\newblock \showarticletitle{Model minimization in Markov decision processes}.
  In \bibinfo{booktitle}{\emph{AAAI/IAAI}}. \bibinfo{pages}{106--111}.
\newblock


\bibitem[\protect\citeauthoryear{Deuber, Egger, Fech, Malavolta, Schr{\"o}der,
  Thyagarajan, Battke, and Durand}{Deuber et~al\mbox{.}}{2019}]%
        {deuber2019my}
\bibfield{author}{\bibinfo{person}{Dominic Deuber}, \bibinfo{person}{Christoph
  Egger}, \bibinfo{person}{Katharina Fech}, \bibinfo{person}{Giulio Malavolta},
  \bibinfo{person}{Dominique Schr{\"o}der}, \bibinfo{person}{Sri
  Aravinda~Krishnan Thyagarajan}, \bibinfo{person}{Florian Battke}, {and}
  \bibinfo{person}{Claudia Durand}.} \bibinfo{year}{2019}\natexlab{}.
\newblock \showarticletitle{My Genome Belongs to Me: Controlling Third Party
  Computation on Genomic Data}.
\newblock \bibinfo{journal}{\emph{Proceedings on Privacy Enhancing
  Technologies}} \bibinfo{volume}{2019}, \bibinfo{number}{1}
  (\bibinfo{year}{2019}), \bibinfo{pages}{108--132}.
\newblock


\bibitem[\protect\citeauthoryear{Deznabi, Mobayen, Jafari, Tastan, and
  Ayday}{Deznabi et~al\mbox{.}}{2018}]%
        {Khodam}
\bibfield{author}{\bibinfo{person}{Iman Deznabi}, \bibinfo{person}{Mohammad
  Mobayen}, \bibinfo{person}{Nazanin Jafari}, \bibinfo{person}{Oznur Tastan},
  {and} \bibinfo{person}{Erman Ayday}.} \bibinfo{year}{2018}\natexlab{}.
\newblock \showarticletitle{An inference attack on genomic data using kinship,
  complex correlations, and phenotype information}.
\newblock \bibinfo{journal}{\emph{IEEE/ACM Transactions on Computational
  Biology and Bioinformatics (TCBB)}} \bibinfo{volume}{15}, \bibinfo{number}{4}
  (\bibinfo{year}{2018}), \bibinfo{pages}{1333--1343}.
\newblock


\bibitem[\protect\citeauthoryear{Duchi, Jordan, and Wainwright}{Duchi
  et~al\mbox{.}}{2013}]%
        {duchi2013local}
\bibfield{author}{\bibinfo{person}{John~C Duchi}, \bibinfo{person}{Michael~I
  Jordan}, {and} \bibinfo{person}{Martin~J Wainwright}.}
  \bibinfo{year}{2013}\natexlab{}.
\newblock \showarticletitle{Local privacy and statistical minimax rates}. In
  \bibinfo{booktitle}{\emph{2013 IEEE 54th Annual Symposium on Foundations of
  Computer Science}}. IEEE, \bibinfo{pages}{429--438}.
\newblock


\bibitem[\protect\citeauthoryear{Dwork}{Dwork}{2008}]%
        {privacy:differentialprivacy}
\bibfield{author}{\bibinfo{person}{Cynthia Dwork}.}
  \bibinfo{year}{2008}\natexlab{}.
\newblock \showarticletitle{Differential privacy: A survey of results}. In
  \bibinfo{booktitle}{\emph{International Conference on Theory and Applications
  of Models of Computation}}. Springer, \bibinfo{pages}{1--19}.
\newblock


\bibitem[\protect\citeauthoryear{Erlingsson, Feldman, Mironov, Raghunathan,
  Talwar, and Thakurta}{Erlingsson et~al\mbox{.}}{2019}]%
        {erlingsson2019amplification}
\bibfield{author}{\bibinfo{person}{{\'U}lfar Erlingsson},
  \bibinfo{person}{Vitaly Feldman}, \bibinfo{person}{Ilya Mironov},
  \bibinfo{person}{Ananth Raghunathan}, \bibinfo{person}{Kunal Talwar}, {and}
  \bibinfo{person}{Abhradeep Thakurta}.} \bibinfo{year}{2019}\natexlab{}.
\newblock \showarticletitle{Amplification by shuffling: From local to central
  differential privacy via anonymity}. In \bibinfo{booktitle}{\emph{Proceedings
  of the Thirtieth Annual ACM-SIAM Symposium on Discrete Algorithms}}. SIAM,
  \bibinfo{pages}{2468--2479}.
\newblock


\bibitem[\protect\citeauthoryear{Fienberg, Slavkovic, and Uhler}{Fienberg
  et~al\mbox{.}}{2011}]%
        {differential:gwas}
\bibfield{author}{\bibinfo{person}{Stephen~E Fienberg},
  \bibinfo{person}{Aleksandra Slavkovic}, {and} \bibinfo{person}{Caroline
  Uhler}.} \bibinfo{year}{2011}\natexlab{}.
\newblock \showarticletitle{Privacy preserving GWAS data sharing}. In
  \bibinfo{booktitle}{\emph{Data Mining Workshops (ICDMW), 2011 IEEE 11th
  International Conference on}}. IEEE, \bibinfo{pages}{628--635}.
\newblock


\bibitem[\protect\citeauthoryear{Gu, Li, Xiong, and Cao}{Gu
  et~al\mbox{.}}{2019}]%
        {gu2019providing}
\bibfield{author}{\bibinfo{person}{Xiaolan Gu}, \bibinfo{person}{Ming Li},
  \bibinfo{person}{Li Xiong}, {and} \bibinfo{person}{Yang Cao}.}
  \bibinfo{year}{2019}\natexlab{}.
\newblock \showarticletitle{Providing Input-Discriminative Protection for Local
  Differential Privacy}.
\newblock \bibinfo{journal}{\emph{arXiv preprint arXiv:1911.01402}}
  (\bibinfo{year}{2019}).
\newblock


\bibitem[\protect\citeauthoryear{Homer, Szelinger, Redman, Duggan, Tembe,
  Muehling, Pearson, Stephan, Nelson, and Craig}{Homer et~al\mbox{.}}{2008}]%
        {related:homer}
\bibfield{author}{\bibinfo{person}{Nils Homer}, \bibinfo{person}{Szabolcs
  Szelinger}, \bibinfo{person}{Margot Redman}, \bibinfo{person}{David Duggan},
  \bibinfo{person}{Waibhav Tembe}, \bibinfo{person}{Jill Muehling},
  \bibinfo{person}{John~V Pearson}, \bibinfo{person}{Dietrich~A Stephan},
  \bibinfo{person}{Stanley~F Nelson}, {and} \bibinfo{person}{David~W Craig}.}
  \bibinfo{year}{2008}\natexlab{}.
\newblock \showarticletitle{Resolving individuals contributing trace amounts of
  DNA to highly complex mixtures using high-density SNP genotyping
  microarrays}.
\newblock \bibinfo{journal}{\emph{PLoS genetics}} \bibinfo{volume}{4},
  \bibinfo{number}{8} (\bibinfo{year}{2008}), \bibinfo{pages}{e1000167}.
\newblock


\bibitem[\protect\citeauthoryear{Humbert, Ayday, Hubaux, and Telenti}{Humbert
  et~al\mbox{.}}{2013}]%
        {genomic:lacks}
\bibfield{author}{\bibinfo{person}{Mathias Humbert}, \bibinfo{person}{Erman
  Ayday}, \bibinfo{person}{Jean-Pierre Hubaux}, {and} \bibinfo{person}{Amalio
  Telenti}.} \bibinfo{year}{2013}\natexlab{}.
\newblock \showarticletitle{Addressing the concerns of the lacks family:
  quantification of kin genomic privacy}. In
  \bibinfo{booktitle}{\emph{Proceedings of the 2013 ACM SIGSAC conference on
  Computer \& communications security}}. ACM, \bibinfo{pages}{1141--1152}.
\newblock


\bibitem[\protect\citeauthoryear{Humbert, Ayday, Hubaux, and Telenti}{Humbert
  et~al\mbox{.}}{2014}]%
        {genomic:reconciling}
\bibfield{author}{\bibinfo{person}{Mathias Humbert}, \bibinfo{person}{Erman
  Ayday}, \bibinfo{person}{Jean-Pierre Hubaux}, {and} \bibinfo{person}{Amalio
  Telenti}.} \bibinfo{year}{2014}\natexlab{}.
\newblock \showarticletitle{Reconciling utility with privacy in genomics}. In
  \bibinfo{booktitle}{\emph{Proceedings of the 13th Workshop on Privacy in the
  Electronic Society}}. ACM, \bibinfo{pages}{11--20}.
\newblock


\bibitem[\protect\citeauthoryear{Johnson and Shmatikov}{Johnson and
  Shmatikov}{2013}]%
        {differential:gwas_johnson}
\bibfield{author}{\bibinfo{person}{Aaron Johnson} {and} \bibinfo{person}{Vitaly
  Shmatikov}.} \bibinfo{year}{2013}\natexlab{}.
\newblock \showarticletitle{Privacy-preserving data exploration in genome-wide
  association studies}. In \bibinfo{booktitle}{\emph{Proceedings of the 19th
  ACM SIGKDD international conference on Knowledge discovery and data mining}}.
  ACM, \bibinfo{pages}{1079--1087}.
\newblock


\bibitem[\protect\citeauthoryear{Kairouz, Oh, and Viswanath}{Kairouz
  et~al\mbox{.}}{2014}]%
        {kairouz2014extremal}
\bibfield{author}{\bibinfo{person}{Peter Kairouz}, \bibinfo{person}{Sewoong
  Oh}, {and} \bibinfo{person}{Pramod Viswanath}.}
  \bibinfo{year}{2014}\natexlab{}.
\newblock \showarticletitle{Extremal mechanisms for local differential
  privacy}. In \bibinfo{booktitle}{\emph{Advances in neural information
  processing systems}}. \bibinfo{pages}{2879--2887}.
\newblock


\bibitem[\protect\citeauthoryear{Liu, Chakraborty, and Mittal}{Liu
  et~al\mbox{.}}{2016}]%
        {liu2016dependence}
\bibfield{author}{\bibinfo{person}{Changchang Liu}, \bibinfo{person}{Supriyo
  Chakraborty}, {and} \bibinfo{person}{Prateek Mittal}.}
  \bibinfo{year}{2016}\natexlab{}.
\newblock \showarticletitle{Dependence Makes You Vulnberable: Differential
  Privacy Under Dependent Tuples.}. In \bibinfo{booktitle}{\emph{NDSS}},
  Vol.~\bibinfo{volume}{16}. \bibinfo{pages}{21--24}.
\newblock


\bibitem[\protect\citeauthoryear{Murakami and Kawamoto}{Murakami and
  Kawamoto}{2019}]%
        {murakami2019utility}
\bibfield{author}{\bibinfo{person}{Takao Murakami} {and}
  \bibinfo{person}{Yusuke Kawamoto}.} \bibinfo{year}{2019}\natexlab{}.
\newblock \showarticletitle{Utility-optimized local differential privacy
  mechanisms for distribution estimation}. In \bibinfo{booktitle}{\emph{28th
  $\{$USENIX$\}$ Security Symposium ($\{$USENIX$\}$ Security 19)}}.
  \bibinfo{pages}{1877--1894}.
\newblock


\bibitem[\protect\citeauthoryear{Naveed, Ayday, Clayton, Fellay, Gunter,
  Hubaux, Malin, and Wang}{Naveed et~al\mbox{.}}{2015}]%
        {survey:genomicera}
\bibfield{author}{\bibinfo{person}{Muhammad Naveed}, \bibinfo{person}{Erman
  Ayday}, \bibinfo{person}{Ellen~W Clayton}, \bibinfo{person}{Jacques Fellay},
  \bibinfo{person}{Carl~A Gunter}, \bibinfo{person}{Jean-Pierre Hubaux},
  \bibinfo{person}{Bradley~A Malin}, {and} \bibinfo{person}{XiaoFeng Wang}.}
  \bibinfo{year}{2015}\natexlab{}.
\newblock \showarticletitle{Privacy in the genomic era}.
\newblock \bibinfo{journal}{\emph{ACM Computing Surveys (CSUR)}}
  \bibinfo{volume}{48}, \bibinfo{number}{1} (\bibinfo{year}{2015}),
  \bibinfo{pages}{6}.
\newblock


\bibitem[\protect\citeauthoryear{Papadimitriou and Tsitsiklis}{Papadimitriou
  and Tsitsiklis}{1987}]%
        {papadimitriou1987complexity}
\bibfield{author}{\bibinfo{person}{Christos~H Papadimitriou} {and}
  \bibinfo{person}{John~N Tsitsiklis}.} \bibinfo{year}{1987}\natexlab{}.
\newblock \showarticletitle{The complexity of Markov decision processes}.
\newblock \bibinfo{journal}{\emph{Mathematics of operations research}}
  \bibinfo{volume}{12}, \bibinfo{number}{3} (\bibinfo{year}{1987}),
  \bibinfo{pages}{441--450}.
\newblock


\bibitem[\protect\citeauthoryear{Samani, Huang, Ayday, Elliot, Fellay, Hubaux,
  and Kutalik}{Samani et~al\mbox{.}}{2015}]%
        {genomic:highorder}
\bibfield{author}{\bibinfo{person}{Sahel~Shariati Samani},
  \bibinfo{person}{Zhicong Huang}, \bibinfo{person}{Erman Ayday},
  \bibinfo{person}{Mark Elliot}, \bibinfo{person}{Jacques Fellay},
  \bibinfo{person}{Jean-Pierre Hubaux}, {and} \bibinfo{person}{Zolt{\'a}n
  Kutalik}.} \bibinfo{year}{2015}\natexlab{}.
\newblock \showarticletitle{Quantifying genomic privacy via inference attack
  with high-order SNV correlations}. In \bibinfo{booktitle}{\emph{Security and
  Privacy Workshops (SPW), 2015 IEEE}}. IEEE, \bibinfo{pages}{32--40}.
\newblock


\bibitem[\protect\citeauthoryear{Shringarpure and Bustamante}{Shringarpure and
  Bustamante}{2015}]%
        {related:shringarpureandbastumante}
\bibfield{author}{\bibinfo{person}{Suyash~S Shringarpure} {and}
  \bibinfo{person}{Carlos~D Bustamante}.} \bibinfo{year}{2015}\natexlab{}.
\newblock \showarticletitle{Privacy risks from genomic data-sharing beacons}.
\newblock \bibinfo{journal}{\emph{The American Journal of Human Genetics}}
  \bibinfo{volume}{97}, \bibinfo{number}{5} (\bibinfo{year}{2015}),
  \bibinfo{pages}{631--646}.
\newblock


\bibitem[\protect\citeauthoryear{Slatkin}{Slatkin}{2008}]%
        {slatkin2008linkage}
\bibfield{author}{\bibinfo{person}{Montgomery Slatkin}.}
  \bibinfo{year}{2008}\natexlab{}.
\newblock \showarticletitle{Linkage disequilibrium—understanding the
  evolutionary past and mapping the medical future}.
\newblock \bibinfo{journal}{\emph{Nature Reviews Genetics}}
  \bibinfo{volume}{9}, \bibinfo{number}{6} (\bibinfo{year}{2008}),
  \bibinfo{pages}{477--485}.
\newblock


\bibitem[\protect\citeauthoryear{Song, Wang, and Chaudhuri}{Song
  et~al\mbox{.}}{2017}]%
        {song2017pufferfish}
\bibfield{author}{\bibinfo{person}{Shuang Song}, \bibinfo{person}{Yizhen Wang},
  {and} \bibinfo{person}{Kamalika Chaudhuri}.} \bibinfo{year}{2017}\natexlab{}.
\newblock \showarticletitle{Pufferfish privacy mechanisms for correlated data}.
  In \bibinfo{booktitle}{\emph{Proceedings of the 2017 ACM International
  Conference on Management of Data}}. ACM, \bibinfo{pages}{1291--1306}.
\newblock


\bibitem[\protect\citeauthoryear{Sutton and Barto}{Sutton and Barto}{2018}]%
        {sutton2018reinforcement}
\bibfield{author}{\bibinfo{person}{Richard~S Sutton} {and}
  \bibinfo{person}{Andrew~G Barto}.} \bibinfo{year}{2018}\natexlab{}.
\newblock \bibinfo{booktitle}{\emph{Reinforcement learning: An introduction}}.
\newblock \bibinfo{publisher}{MIT press}.
\newblock


\bibitem[\protect\citeauthoryear{Wagner}{Wagner}{2017}]%
        {wagner2017evaluating}
\bibfield{author}{\bibinfo{person}{Isabel Wagner}.}
  \bibinfo{year}{2017}\natexlab{}.
\newblock \showarticletitle{Evaluating the strength of genomic privacy
  metrics}.
\newblock \bibinfo{journal}{\emph{ACM Transactions on Privacy and Security
  (TOPS)}} \bibinfo{volume}{20}, \bibinfo{number}{1} (\bibinfo{year}{2017}),
  \bibinfo{pages}{1--34}.
\newblock


\bibitem[\protect\citeauthoryear{Wang, Li, Wang, Tang, and Zhou}{Wang
  et~al\mbox{.}}{2009}]%
        {related:wang}
\bibfield{author}{\bibinfo{person}{Rui Wang}, \bibinfo{person}{Yong~Fuga Li},
  \bibinfo{person}{XiaoFeng Wang}, \bibinfo{person}{Haixu Tang}, {and}
  \bibinfo{person}{Xiaoyong Zhou}.} \bibinfo{year}{2009}\natexlab{}.
\newblock \showarticletitle{Learning your identity and disease from research
  papers: information leaks in genome wide association study}. In
  \bibinfo{booktitle}{\emph{Proceedings of the 16th ACM conference on Computer
  and communications security}}. ACM, \bibinfo{pages}{534--544}.
\newblock


\bibitem[\protect\citeauthoryear{Wang, Blocki, Li, and Jha}{Wang
  et~al\mbox{.}}{2017}]%
        {wang2017locally}
\bibfield{author}{\bibinfo{person}{Tianhao Wang}, \bibinfo{person}{Jeremiah
  Blocki}, \bibinfo{person}{Ninghui Li}, {and} \bibinfo{person}{Somesh Jha}.}
  \bibinfo{year}{2017}\natexlab{}.
\newblock \showarticletitle{Locally differentially private protocols for
  frequency estimation}. In \bibinfo{booktitle}{\emph{26th $\{$USENIX$\}$
  Security Symposium ($\{$USENIX$\}$ Security 17)}}. \bibinfo{pages}{729--745}.
\newblock


\bibitem[\protect\citeauthoryear{Wang, Li, and Jha}{Wang et~al\mbox{.}}{2018}]%
        {wang2018locally}
\bibfield{author}{\bibinfo{person}{Tianhao Wang}, \bibinfo{person}{Ninghui Li},
  {and} \bibinfo{person}{Somesh Jha}.} \bibinfo{year}{2018}\natexlab{}.
\newblock \showarticletitle{Locally differentially private frequent itemset
  mining}. In \bibinfo{booktitle}{\emph{2018 IEEE Symposium on Security and
  Privacy (SP)}}. IEEE, \bibinfo{pages}{127--143}.
\newblock


\bibitem[\protect\citeauthoryear{Wang, Huang, Zhao, Tang, Wang, and Bu}{Wang
  et~al\mbox{.}}{2015}]%
        {related:wangPrivateEditDistance}
\bibfield{author}{\bibinfo{person}{Xiao~Shaun Wang}, \bibinfo{person}{Yan
  Huang}, \bibinfo{person}{Yongan Zhao}, \bibinfo{person}{Haixu Tang},
  \bibinfo{person}{XiaoFeng Wang}, {and} \bibinfo{person}{Diyue Bu}.}
  \bibinfo{year}{2015}\natexlab{}.
\newblock \showarticletitle{Efficient genome-wide, privacy-preserving similar
  patient query based on private edit distance}. In
  \bibinfo{booktitle}{\emph{Proceedings of the 22nd ACM SIGSAC Conference on
  Computer and Communications Security}}. ACM, \bibinfo{pages}{492--503}.
\newblock


\bibitem[\protect\citeauthoryear{Warner}{Warner}{1965}]%
        {warner1965randomized}
\bibfield{author}{\bibinfo{person}{Stanley~L Warner}.}
  \bibinfo{year}{1965}\natexlab{}.
\newblock \showarticletitle{Randomized response: A survey technique for
  eliminating evasive answer bias}.
\newblock \bibinfo{journal}{\emph{J. Amer. Statist. Assoc.}}
  \bibinfo{volume}{60}, \bibinfo{number}{309} (\bibinfo{year}{1965}),
  \bibinfo{pages}{63--69}.
\newblock


\bibitem[\protect\citeauthoryear{Yang, Sato, and Nakagawa}{Yang
  et~al\mbox{.}}{2015}]%
        {yang2015bayesian}
\bibfield{author}{\bibinfo{person}{Bin Yang}, \bibinfo{person}{Issei Sato},
  {and} \bibinfo{person}{Hiroshi Nakagawa}.} \bibinfo{year}{2015}\natexlab{}.
\newblock \showarticletitle{Bayesian differential privacy on correlated data}.
  In \bibinfo{booktitle}{\emph{Proceedings of the 2015 ACM SIGMOD international
  conference on Management of Data}}. ACM, \bibinfo{pages}{747--762}.
\newblock


\bibitem[\protect\citeauthoryear{Yu, Fienberg, Slavkovi{\'c}, and Uhler}{Yu
  et~al\mbox{.}}{2014}]%
        {differential:gwas_yu}
\bibfield{author}{\bibinfo{person}{Fei Yu}, \bibinfo{person}{Stephen~E
  Fienberg}, \bibinfo{person}{Aleksandra~B Slavkovi{\'c}}, {and}
  \bibinfo{person}{Caroline Uhler}.} \bibinfo{year}{2014}\natexlab{}.
\newblock \showarticletitle{Scalable privacy-preserving data sharing
  methodology for genome-wide association studies}.
\newblock \bibinfo{journal}{\emph{Journal of biomedical informatics}}
  \bibinfo{volume}{50} (\bibinfo{year}{2014}), \bibinfo{pages}{133--141}.
\newblock


\end{thebibliography}

\appendix

\vspace{5mm}
\noindent{\Large{\textbf{Appendices}}}
\normalsize

\section{Flowchart of Probability Distribution Scenarios for SNP Selection}
\label{app:tree}

We give the flowchart of the selection of next SNP to be processed by the data sharing algorithm in Figure \ref{fig:treei}.

\begin{figure*}[h]
    \centering
    \includegraphics[width=1\textwidth]{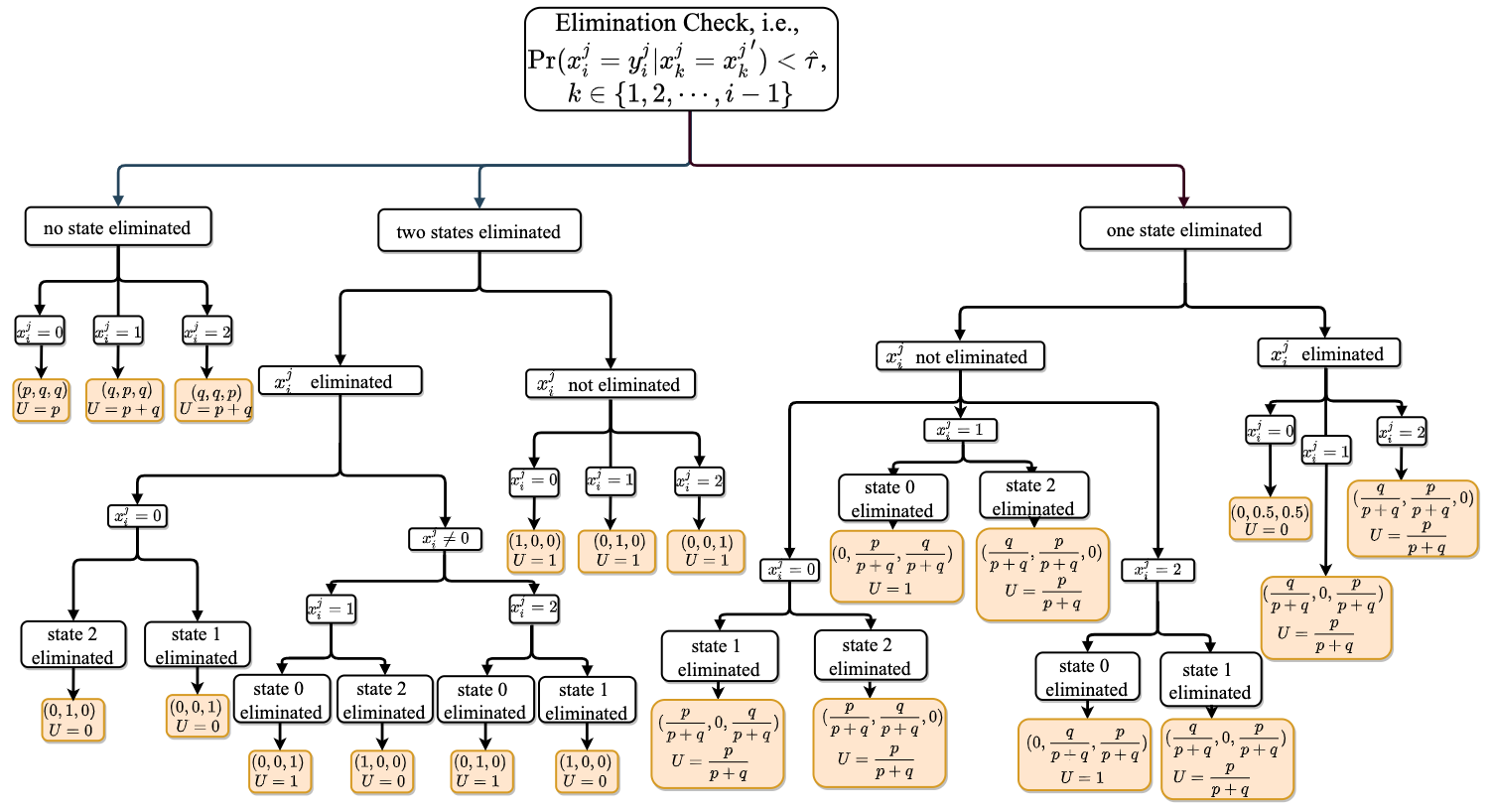}
    \caption{Tree structured flowchart when individual $I_j$'s SNP $x_i^j$ is selected in Algorithm~\ref{alg:alg1}. Specifically, the root node calculates all possible probabilities, i.e., line 5-6 in Algorithm~\ref{alg:alg1}, the other non-leaf nodes conduct state eliminations, i.e., line 14-19 in Algorithm~\ref{alg:alg1}, and the leaf nodes (in yellow) share the non-eliminated states, i.e., line 20 in Algorithm~\ref{alg:alg1}. For each $x_i^j$, only one leaf will be activated. The tuple in a leaf node is the probability distribution of sharing different states, i.e., Figure~\ref{fig:prob2}, and $U$ represents the immediate expected value of the utility of that leaf being activated.}
    \label{fig:treei}
\end{figure*}

\section{Proof of Lemma 5.1}
\label{app:proof}

\begin{proof}[Proof of Lemma \ref{individual_utility}]
Considering the SNP $i$ in a population of $n$ individuals, i.e., $\{I_1,I_2,\cdots,I_n\}$, we denote the original and shared SNP $i$ for the individuals as $\{x_i^1,x_i^2,\cdots,x_i^n\}$ and $\{y_i^1,y_i^2,\cdots,y_i^n\}$, respectively. Then, the expected beacon utility of the population is $\mathbb{U}_{\text{pop}} = 1\times\Pr(\sum_jy_i^j=0|\sum_jx_i^j=0)+1\times\Pr(\sum_jy_i^j\neq0|\sum_jx_i^j\neq0),  j\in\{1,2,\cdots,n\}$, and the expected beacon utility of an individual $I_j$ is $\mathbb{U}_{I_j} = 1\times\Pr(y_i^j=0|x_i^j=0)+1\times\Pr(y_i^j\neq0|x_i^j\neq0)$.

Since $x_i^j,y_i^j\in\{0,1,2\}$, for the first part of $\mathbb{U}_{\text{pop}}$, we have $\Pr(\sum_jy_i^j=0|\sum_jx_i^j=0)\leq\sum_j \Pr(y_i^j=0|x_i^j=0)$. For the second part, we have
\begin{equation*}
\begin{aligned}
&\Pr(\sum_jy_i^j\neq0|\sum_jx_i^j\neq0)  \\
=& \Pr(\exists y_i^j\neq 0|\exists x_i^k\neq 0) \qquad\Big( j,k\in\{1,2,\cdots,n\}\Big)\\
\leq &\sum_{j,k} \Pr( y_i^j\neq 0| x_i^k\neq 0) \qquad\Big( j,k\in\{1,2,\cdots,n\}\Big)\\
\stackrel{*}=& \sum_{j} \Pr( y_i^j\neq 0| x_i^j\neq 0),
\end{aligned}
\end{equation*}
where $*$ is due to the fact that the individuals do not have access to others' SNPs and they share SNPs independently from each other. Thus, we have
$\mathbb{U}_{\text{pop}} \leq \sum_j \Big(\Pr(y_i^j=0|x_i^j=0)+\Pr(y_i^j\neq0|x_i^j\neq0)\Big) = \sum_j\mathbb{U}_{I_j}$.
Let $\mathbb{U}_{I_j}^*$ be the maximum of the expectation of individual $I_j$'s beacon utility. Then, we can obtain the maximum of the expectation of a population's beacon utility as $\mathbb{U}_{\text{pop}}^* =  \max\{\sum_j \mathbb{U}_{I_j}^*,1\}$, which suggests that the sufficiency holds.
\end{proof}

\section{The Accuracy of Beacon Responses for Both Yes and No Answers}
\label{app:accuracy}

In Figure~\ref{fig:beaconYesNo} (Section \ref{subsec:compareLDP}), we show the utility of the proposed mechanism in terms of the accuracy of the beacon for 1000 queries. Here, we give the accuracy for both ``yes'' and ``no'' separately. In beacon responses, it is more difficult to preserve the accuracy after perturbation when the original response of a query is ``no'', because to provide an accurate response, all collected values of a SNP should still be $0$ after perturbation. When there is at least one other value (1 or 2), the response changes and becomes incorrect. The proposed scheme significantly outperforms the original RR mechanism for the accuracy of ``no'' responses. Since the minor allele is observed rarely in the population, most of the individuals have $0$ as their SNP values (i.e., most of their SNPs contain no minor alleles). Also, since the proposed scheme considers correlations, it mostly eliminates states $1$ and/or $2$ for such SNPs due to their low frequency in the population. Hence, when there is no minor alleles in the individuals, with high probability, all individuals report $0$ for that SNP, which preserves the utility of the beacon responses. Our results are shown in Figure~\ref{fig:beaconNo} (for the accuracy of only ``no'' responses) and Figure~\ref{fig:beaconYes} (for the accuracy of only ``yes'' responses). For $\epsilon \leq 1$, the proposed scheme also provides better accuracy than the original RR mechanism when the original response of a query is ``yes''. When $\epsilon > 1$, the accuracy of `yes'' responses is similar for both mechanisms.

\begin{figure}[h]
\centering
\includegraphics[width=8cm,keepaspectratio]{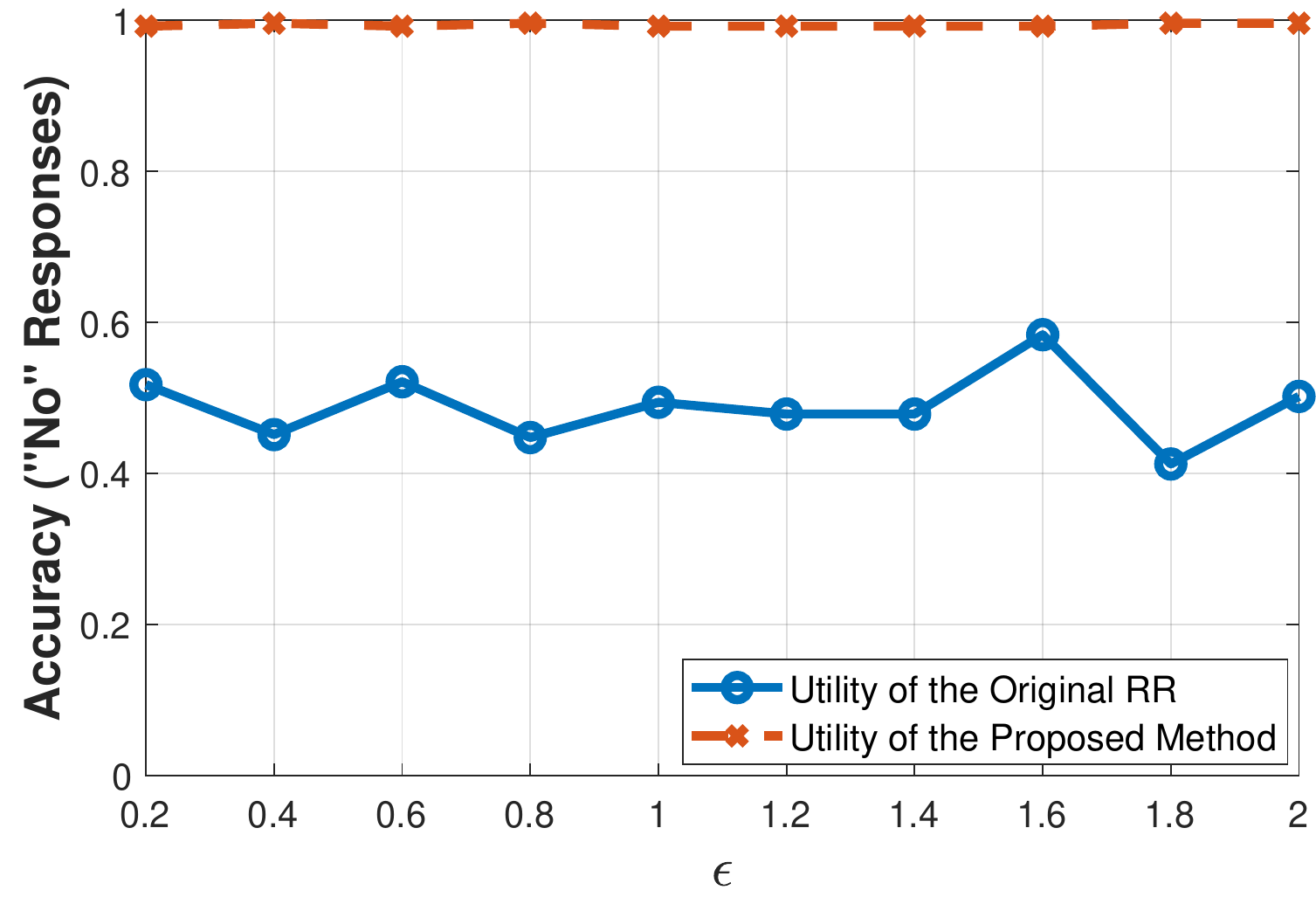}
\caption{Comparison of the proposed method with original randomized response mechanism in terms of utility. Utility is measured as the accuracy of ``no'' responses provided from a genomic data sharing beacon.}
\label{fig:beaconNo}
\end{figure}

\begin{figure}[h]
\centering
\includegraphics[width=8cm,keepaspectratio]{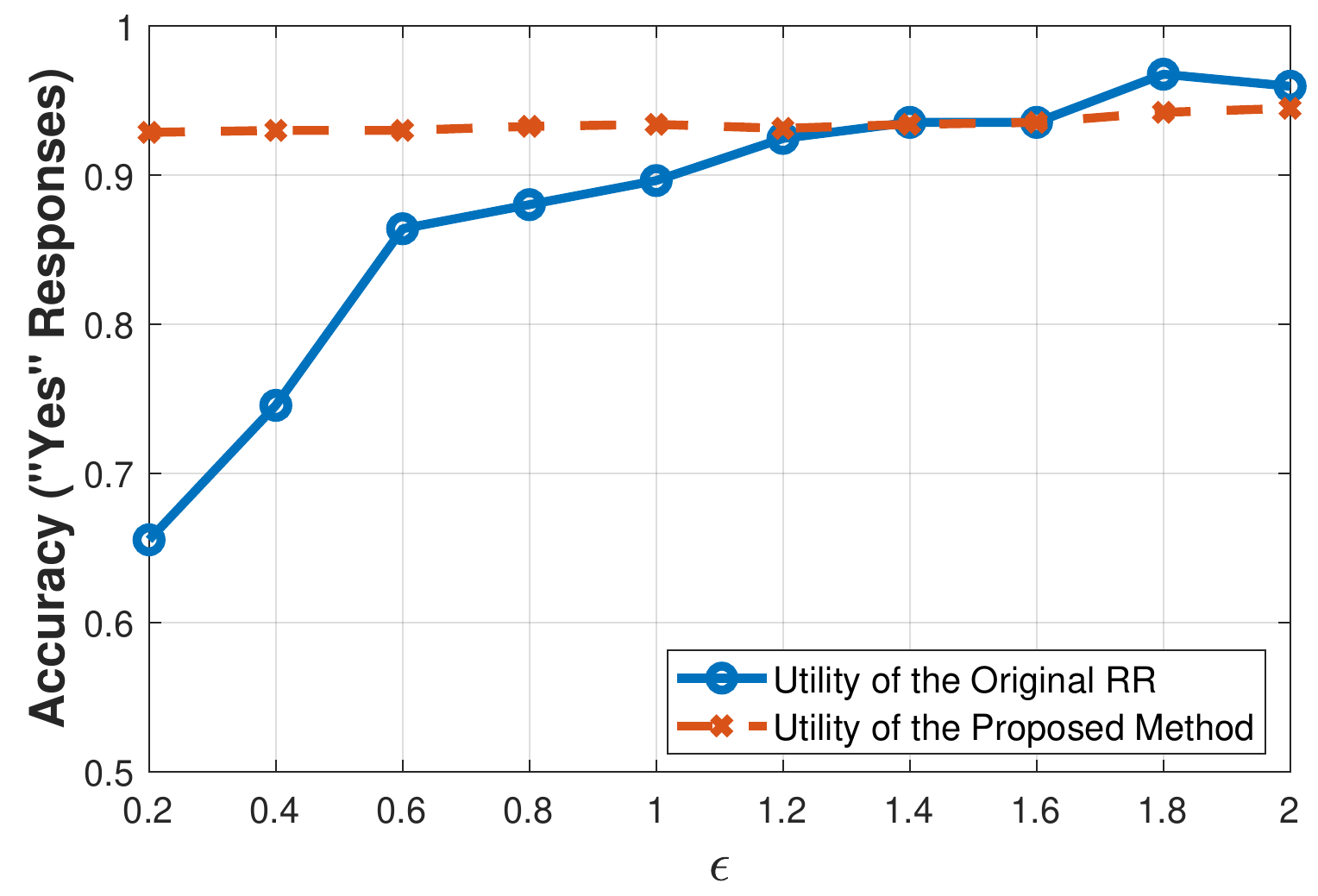}
\caption{Comparison of the proposed method with original randomized response mechanism in terms of utility. Utility is measured as the accuracy of ``yes'' responses provided from a genomic data sharing beacon.}
\label{fig:beaconYes}
\end{figure}

\section{Comparing Utility with the Original RR Mechanism after Post-Processing}
\label{app:postprocess}

In Table \ref{table:postprocess}, we give the utility comparison of the proposed mechanism with RR mechanism when post-processing is applied. 

\begin{table}
\small
\centering
\caption{Comparison of the proposed method with original randomized response mechanism (with and without post-processing) in terms of utility. Utility is measured as the accuracy of responses provided from a genomic data sharing beacon.}
\begin{tabular}{ c c c c c c }
\hline
& \multicolumn{5}{c}{$\epsilon$} \\
\hline
 & 0.4 & 0.8 & 1.2 & 1.6 & 2 \\
\hline
RR without post-processing & 0.697 & 0.791 & 0.808 & 0.831 & 0.849 \\
\hline
RR with post-processing & 0.667 & 0.715 & 0.734 & 0.765 & 0.788 \\
\hline
Proposed method & 0.934 & 0.941 & 0.945 & 0.952 & 0.961 \\
\hline
\end{tabular}
\label{table:postprocess}
\end{table}

\section{The Effect of Sharing a SNP by Children on a Parent}
\label{app:kinship}

In Section \ref{sec:discussion}, we discuss how data sharing affects the privacy of family members. Here, we analyze how the privacy of a parent is affected when one or more of his/her children share their SNPs.

\subsection{Sharing a SNP by One Child}
\label{app:kinship1}
In this section, we first analyze the privacy loss of a victim when one of her/his first degree relatives share her $SNP_i$ under $(\epsilon,T)$-dependent LDP. For this, we consider a parent as the victim (assuming that victim did not share genomic data before) and his/her child as the donor that shares her genome.

Let $I_f$ be a parent of the donor $I_j$. In Table~\ref{table:mendel}, we show the probability of having each SNP value for the parent when the child's SNP value is known. For instance, the attacker can compute the probability of the parent $I_f$ having each possible value for $SNP_i$ when that SNP is shared as $0$ by his/her child $I_j$ under $(\epsilon,T)$-dependent LDP as follows:
\begin{displaymath}
\Pr(x_{i}^f = 0~|~y_i^j = 0) = \frac{2}{3}p + \frac{1}{3}q + 0q
\end{displaymath}
\begin{displaymath}
\Pr(x_{i}^f = 1~|~y_i^j = 0) = \frac{1}{3}p + \frac{1}{3}q + \frac{1}{3}q
\end{displaymath}
\begin{displaymath}
\Pr(x_{i}^f = 2~|~y_i^j = 0) = 0p + \frac{1}{3}q + \frac{2}{3}q
\end{displaymath}

To quantify the information gain of the attacker about the parent $I_f$, we compute the indirect privacy budget ($\hat{\epsilon_{f}}$) of the parent $I_f$ as a result of data sharing by the child. Let the child shares her $SNP_i$ with her privacy budget $\epsilon_j$. From the definition of LDP, $p \geq q$ since $\epsilon_j \geq 0$, and the equality holds when $\epsilon_j = 0$. Hence, $\Pr(x_{i}^f = 0~|~y_i^j = 0) \geq \Pr(x_{i}^f = 1~|~y_i^j = 0) \geq \Pr(x_{i}^f = 2~|~y_i^j = 0)$. To satisfy $(\epsilon,T)$-dependent LDP, $\frac{\Pr(x_{i}^f = 0~|~y_i^j = 0)}{\Pr(x_{i}^f = 0~|~y_i^j = 2)} \leq e^{\hat{\epsilon_{f}}}$. Using the equations above, we obtain $\frac{(2p + q)/3}{q} \leq e^{\hat{\epsilon_{f}}}$. Therefore, the indirect privacy budget of the parent $\hat{\epsilon_{f}} = \ln{(\frac{2e^{\epsilon_j} + 1}{3})}$ when the child shares $0$ as the value of one of her SNPs (i.e., $y_i^j = 0$). Since $e^{\epsilon_j} \geq 1$ we obtain $\hat{\epsilon_{f}} \leq \epsilon_j$ and equality holds when $\epsilon_j = 0$. Due to symmetric probability distributions, we also obtain the same $\hat{\epsilon_{f}}$ when $y_i^j = 2$. 
When the child shares $1$ as the value of $SNP_i$ under $(\epsilon,T)$-dependent LDP ($y_i^j = 1$), all $\Pr(x_{i}^f = 0~|~y_i^j = 1)$, $\Pr(x_{i}^f = 1~|~y_i^j = 1)$, and $\Pr(x_{i}^f = 2~|~y_i^j = 1)$ values are computed as $\frac{1}{3}p + \frac{2}{3}q$. Hence, we calculate the indirect privacy budget of the parents as $\hat{\epsilon_{f}} = 0$ in that case, which means that the attacker does not gain any information (about the parent) by receiving $1$ as the perturbed SNP value (of the child) under $(\epsilon,T)$-dependent LDP. 

\begin{table}
\small
	\caption{Probability distribution of SNP values of a parent given the SNP values of the child based on Mendel's law. $x_i^j$ is $SNP_i$ of the child $I_j$ and $x_i^f$ is the $SNP_i$ of the parent $I_f$.}
		\label{table:mendel}
		\centering
		\begin{tabular}{c|c|c|c|c}
			\hline
		\multicolumn{2}{c|}{} & \multicolumn{3}{c}{Parent ($I_f$)}
		\\ \cline{3-5}
		\multicolumn{2}{c|}{} &  $x_i^f = 0$  & $x_i^f = 1$ & $x_i^f = 2$ 
		\\\hline
		\parbox[t]{2mm}{\multirow{3}{*}{\rotatebox[origin=c]{90}{Child ($I_j$)}}} & $x_i^j = 0$ & $2/3$ & $1/3$  & $0$  \\\cline{2-5}
 &  $x_i^j = 1$  & $1/3$  & $1/3$  & $1/3$ \\\cline{2-5}
 &  $x_i^j = 2$  & $0$ & $1/3$ & $2/3$ \\
 \hline
		\end{tabular}
	\end{table}

In any case, if the privacy budgets of the parent and the child are the same (i.e., $\epsilon_{f} = \epsilon_{j}$), the privacy of the parent is not violated (due to the sharing of the the child) since $\hat{\epsilon_{f}}$ is not greater than the privacy budget of the parent ($\epsilon_{f}$). However, the privacy budget of the parent might be exceeded if the parent has a lower privacy budget than the child. In that case, the child should consider the privacy of her parent and select a $\epsilon_{max_i}$ value (while sharing her data) to preserve the privacy of parent $I_f$. We compute $\epsilon_{max_i} = \ln{(\frac{3e^{\epsilon_{f}} - 1}{2})}$. For instance, if the parent $I_f$ has a privacy budget $\epsilon_{f} = 1$, the child $I_j$ should use a privacy budget less than or equal to $\epsilon_{max_i} = 1.27$ for sharing $SNP_i$. Note that, $\epsilon_{max_i}$ is computed for only one SNP ($SNP_i$) and it can similarly be computed for other SNPs as well. As expected, privacy of a victim degrades even more when several of victim's first degree family members share their SNPs. To show this, we extend the previous discussion and consider a case in which two children of a parent share their SNPs under $(\epsilon,T)$-dependent LDP in Appendix~\ref{app:kinship2}. 

\subsection{Sharing a SNP by Two Children}
\label{app:kinship2}

Let the second child be $I_k$ (sibling of $I_j$) and her privacy budget $\epsilon_k$ be also equal to $\epsilon_j$ for simplicity. To show the degradation of privacy, we compute the indirect privacy budget of the parent $\hat{\epsilon_{f}}$ when both children share $0$ as the value of their $SNP_i$. This can also be shown for other scenarios. When both siblings share $0$, with probability $p^2$ their original SNP values are also $0$. We compute the probability of having other pairs (as the original SNP values of the siblings) similarly. For instance, when both siblings share $0$, with probability $q^2$ their original SNP values are $2$. Using the probabilities in Table~\ref{table:mendel2}, we compute the following probabilities.
\begin{displaymath}
\Pr(x_{i}^P = 0~|~y_i^j = 0 \wedge y_i^k = 0) = \frac{4}{5}p^2 + \frac{4}{5}pq + \frac{5}{13}q^2
\end{displaymath}
\begin{displaymath}
\Pr(x_{i}^P = 1~|~y_i^j = 0 \wedge y_i^k = 0) = \frac{1}{5}p^2 + \frac{16}{5}pq + \frac{106}{65}q^2
\end{displaymath}
\begin{displaymath}
\Pr(x_{i}^P = 2~|~y_i^j = 0 \wedge y_i^k = 0) = 0p^2 + 0pq + \frac{129}{65}q^2
\end{displaymath}

Similar to the computations in Section~\ref{sec:discussion}, 
$e^{\hat{\epsilon_{f}}} = \frac{\Pr(x_{i}^f = 0~|~y_i^j = 0 \wedge y_i^k = 0)}{\Pr(x_{i}^f = 2~|~y_i^j = 0 \wedge y_i^k = 0)}$ $ = \frac{52(e^{\epsilon_j})^2 + 52e^{\epsilon_j} + 25}{129}$. Hence, $\hat{\epsilon_{f}} = \ln{(\frac{52(e^{\epsilon_j})^2 + 52e^{\epsilon_j} + 25}{129})}$. When we solve the quadratic equation $e^{\hat{\epsilon_{j}}} = \frac{52(e^{\epsilon_j})^2 + 52e^{\epsilon_j} + 25}{129}$, we obtain a root at $\epsilon_j = 0$. Hence, $\hat{\epsilon_{f}} \geq \epsilon_j$ and equality holds when $\epsilon_j = 0$. This result clearly shows that when two siblings share their SNPs under $(\epsilon,T)$-dependent LDP with the same privacy budget $\epsilon_j$, the attacker may learn more information about the parent and the privacy of the parent is violated if the privacy budget of the parent is also equal to the privacy budgets of the children.
\begin{table}
\small
	\caption{Probability distribution of SNP value of a parent given the SNP values of two children based on Mendel's law.}
		\label{table:mendel2}
		\centering
		\begin{tabular}{c|c|c|c|c}
			\hline
		\multicolumn{2}{c|}{} & \multicolumn{3}{c}{Parent ($I_f$)}
		\\ \cline{3-5}
		\multicolumn{2}{c|}{} & $x_i^f = 0$ & $x_i^f = 1$ & $x_i^f = 2$ 
		\\\hline
		\parbox[t]{2mm}{\multirow{9}{*}{\rotatebox[origin=c]{90}{Children ($I_j, I_k$)}}} & $x_i^j = 0 \wedge x_i^k = 0$ & $4/5$ & $1/5$  & $0$  \\
		\cline{2-5}
 & $x_i^j = 0 \wedge x_i^k = 1$ & $2/5$  & $3/5$  & $0$ \\
 \cline{2-5}
 & $x_i^j = 0 \wedge x_i^k = 2$ & $0$ & $1$ & $0$ \\
  \cline{2-5}
 & $x_i^j = 1 \wedge x_i^k = 0$ & $2/5$ & $3/5$ & $0$ \\
  \cline{2-5}
 & $x_i^j = 1 \wedge x_i^k = 1$ & $5/13$ & $3/13$ & $5/13$ \\
  \cline{2-5}
 & $x_i^j = 1 \wedge x_i^k = 2$ & $0$ & $3/5$ & $2/5$ \\
  \cline{2-5}
 & $x_i^j = 2 \wedge x_i^k = 0$ & $0$ & $1$ & $0$ \\
  \cline{2-5}
 & $x_i^j = 2 \wedge x_i^k = 1$ & $0$ & $3/5$ & $2/5$ \\
  \cline{2-5}
 & $x_i^j = 2 \wedge x_i^k = 2$ & $0$ & $1/5$ & $4/5$ \\
 \hline
		\end{tabular}
	\end{table}
	
Let all $I_j$, $I_k$, and $I_f$ have equal privacy budgets ($\epsilon = \epsilon_j = \epsilon_k = \epsilon_f $), which they do not want to exceed. Assume the first child shared her SNP with $\epsilon$ and the second child wants to keep the privacy budget of her parent equal to $\epsilon$ while sharing her data. In that case, the second child should use an adjusted privacy budget ($\epsilon_{max_i}$) which is less than $\epsilon$. When we compute $\epsilon_{max_i}$, we obtain $\epsilon_{max_i} = \ln{(\frac{103e^{\epsilon} - 25}{52e^{\epsilon} + 26})}$. Hence, the second child needs to share her SNP $i$ with such $\epsilon_{max_i}$ to keep her parent's privacy budget as $\epsilon$. For instance, if the privacy budget $\epsilon = 0.5$ and the first child shared her SNP with the same $\epsilon$, the second child should select $\epsilon_{max_i}$ as 0.259 to keep her parent's privacy budget as $\epsilon = 0.5$.

\end{document}